\documentclass[a4paper,UKenglish, thm-restate,DIV=10]{article}
\usepackage[a4paper, left=25mm, right=25mm, top=30mm,bottom=30mm]{geometry}
\usepackage{graphicx}
\usepackage{multirow}
\usepackage{booktabs}
\usepackage{color}
\usepackage{capt-of}
\usepackage{subcaption}
\usepackage{float}
\usepackage{stmaryrd}
\usepackage{amsmath}
\usepackage{mathtools}
\usepackage{mathrsfs}
\usepackage{tikz} 
\usepackage{sidecap}
\usepackage[]{algorithm2e}
\SetAlFnt{\small\sffamily}
\usepackage{amssymb}
\usepackage{amsthm}
\usepackage[strings]{underscore}
\usepackage{tikz}
\usepackage{wrapfig}
\usepackage[toc,page]{appendix}
\usetikzlibrary{arrows,automata}
\newtheorem{theorem}{Theorem}[section]
\newtheorem{corollary}{Corollary}[theorem]
\newtheorem{lemma}[theorem]{Lemma}

\newtheorem{definition}{Definition}[section]
\newcommand{\algoname}[1]{\textnormal{\textsc{#1}}}
\newcommand{\algonamebf}[1]{\textnormal{\textsc{\textbf{#1}}}}

\newcommand{\state}[1]{{\small \textsf{#1}}}
\newcommand{\remove}[1]{}
\newcommand{\etal}{et al.}

\newcommand{\ep}{\varepsilon}
\newcommand{\cM}{{\mathcal{M}}}
\newcommand{\cS}{{\mathcal{S}}}
\newcommand{\polylog}{{\text{\,polylog }}}

\providecommand{\keywords}[1]{\textbf{\textit{Index terms---}} #1}

\title{Contention resolution on a restrained channel}

\begin{document}

\author{Elijah~Hradovich, Marek~Klonowski and Dariusz R.~Kowalski}
\date{}


\maketitle

\begin{abstract}%
We examine deterministic contention resolution on a multiple-access channel
when packets are injected continuously by an adversary 
to the buffers of the stations
at rate $\rho$ packages per round.
The aim is to successfully transmit packets and maintain system stability, that is, bounded queues even in infinite perspective of execution.
The largest injection rate for which a given contention resolution algorithm guaranties stability is called (algorithm's) throughput.
In contrast to the previous work, we assume that there is a strict limit $k$ on the total number of stations allowed to transmit or listen to the channel at a given time, that can never be exceeded; we call such channel a {\em $k$-restrained channel}. 
%
%
We construct adaptive and full sensing protocols with throughput $1$ and $1-1/n$,
respectively, in a constant-restrained channel. 
In other words for both classes of algorithms we do not have to sacrifice the throughput, even if the channel is very restrained. 
%
For the case of acknowledgment based algorithms that cannot adapt to the situation on the channel (i.e., their transmission schedules are fixed in advance), we show that a restrained channel causes throughput reduction by factor at least $\min \{ \frac{k}{n}, \frac{1}{3\log n }\}$ and
at most $\Theta(\frac{k}{n \log^2 (n)})$. Our acknowledgment-based algorithm is based on a newly introduced and thoroughly studied, for its own independent interest, 
$k$-light selector. 
%
We support our theoretical analysis by simulation results of algorithms constructed in the paper 
in systems of moderate, realistic sizes and scenarios. We also 
compare our algorithms with backoff algorithms, which are common in real-world channel implementations, in terms of throughput, queue sizes and channel restrain needed for achieving 
their 
throughputs.

\end{abstract}


%
%

\section{Introduction}  
The fundamental problem of access to a single medium by multiple devices 
is faced by many
different kinds of distributed systems, including but not limited to: processor transactions on memory, radio networks communication, services access to a shared resource on machines or data-centers.
The medium constraints a system by collisions or denial of service when more than one device attempts to use it simultaneously.
However, the number of simultaneous attempts also matters, both in practice and, as we will show in this work, in theory.
In practice, channel access is constrained by physical factors, such as power, energy or availability. 
First, the "energy" spent by 
devices during such unsuccessful-for-most attempts 
is usually wasted. 
Second, for a case of multi-hop radio communication, too many attempts to transmit by neighbors may not only cause a collision in the considered node,
but also in nodes of further distance.
Third example,
hardware systems are designed with a spike (maximal) power use in mind to prevent meltdown or blackout. 

The above examples
have led us to an investigation of restrained-channels, as natural extension of the classical shared-channel communication model with $n$  devices (stations) attached to a single communication medium.
Due to constrains of the channel at most one station can successfully transmit a single packet during one round, and an additional
restrain limits the number of simultaneous activities (transmissions or listenings) on the channel.
We focus on dynamic scenario when an adversary injects (in an arbitrary way) at most $0< \rho \leq 1$  packets per round, on average, to stations' buffers. 
The primary goal is to design an algorithm that guarantees \textit{stability}, that is, a property that the sizes of queues in buffers stay bounded for the highest possible injection rate $\rho$, which we will be calling {\em throughput}. Another important aim is to minimize channel restrain to achieve the maximum possible (i.e., in non-restrained classical channel) throughput.
In other words: when and how the "restrain orders"\footnote{%
There is a loose analogy of $k$-restrained channel and limitations on the number of people gathering in publicly available spaces
established by restrained orders during the COVID-19 pandemic.}
reduce dynamic utilization of shared resources.

\paragraph*{{\bf Our contribution}}

In this work we investigate how
limiting the power 
of a distributed scheduler by at most $k$ 
active stations 
per round, which we call $k$-restrained channel, influences the efficiency of the broadcasting system. In other words, what throughput could be achieved on $k$-restrained channels. We focus on deterministic solutions.

We construct optimal or nearly-optimal
solutions for different classes of protocols studied in the literature: 
achieving throughput $1$ for adaptive protocols, 
throughput $1-1/n$ 
for full-sensing protocols, and sub-optimal 
throughput 
$\Theta(\frac{k}{n \log^2 (n)})$ 
for acknowledgement based protocols (the latter result is complemented by the upper bound 
$\min \{ \frac{k}{n}, \frac{1}{3\log n }\}$
for this class of protocols). 
%
The main conclusion from our results is that for some classes, i.e., adaptive and full sensing 
protocols, we are able to construct 
algorithms without decreasing throughput of the system
(i.e., comparing to the corresponding protocols without 
system restraint, for which the upper bound is
$1$ in case of adaptive protocols, while the throughput $1$ is not achievable by full-sensing ones~\cite{kow}). 
In some other classes, e.g., acknowledgement based protocols, 
restraining the channel limits the throughput of efficient solutions. 
Note also that for adaptive algorithms randomization
could not help, as the optimal throughput is achieved for  channel restrain $2$ that is necessary for any communication.
Another consequence is that our adaptive and full sensing algorithms achieve a constant amortized numbers of transmissions/listenings
per packet; the acknowledgment-based solution guarantees $O(n\log^2 n)$ amortized number.

Let us stress that our acknowledgement-based algorithm uses a newly introduced structure, called $k$-light selector, which we thoroughly study
for its own independent interest.

Apart from rigid formal analysis, 
part of this paper is  devoted to experimental results proving the efficiency of the constructed protocols for realistic systems' sizes.
We also show that our algorithms outperform back-off-type protocols both in terms of 
throughput-efficiency  
and system stability (i.e., queue sizes) in the model with
restraint. 

\subsection{Previous and related work} 

In our paper we study distributed broadcasting on multiple access channel in the framework of adversarial queuing and 
limited power available 
to the system. The adversary definition, taxonomy of different models of communication channels (w.r.t. stations capabilities) are the same as~in~\cite{kow} by Chlebus~\etal, wherein the authors considered throughput for different settings without 
limiting the number of active stations.
In another study by Chlebus~\etal~ \cite{Chlebus:2019:EEA:3323165.3323190}, the authors introduced the energy cap, which is equivalent definition to channel restrain concept used in our paper. 
Randomized queue-free constant throughput-based model for packet broadcast with only bounds on transmission energy being known was studied by De~Marco~\etal~\cite{DBLP:MarcoS17}.
Bender~\etal~studied possibility of achieving constant throughput on channels without collision detection in~\cite{bender2020contention}  and presented an algorithm for contention resolution 
achieving constant throughput using only constant {\em average energy} for transmitting in~\cite{BenderKPY16}.
Awerbuch~\etal~accounted for the energy cost of transmissions and developed energy-efficient and jamming-resistant protocol in \cite{Awerbuch:2008:JMP:1400751.1400759}.
Problems of energy-efficient leader election, size approximation and census were studied by Chang~\etal~in~\cite{Kopelowitz17} in single-hop networks obtaining energy-optimal algorithms for different models. 
 Earlier research on energy complexity of leader election and related problems includes Jurdzinski~\etal~\cite{JurdzinskiKZ02, enerJurdziskiKZ03}. Energy-efficient size approximation of a single hop radio networks were analyzed by Jurdzinski~\etal~in~\cite{jkzApprox} and by Kardas~\etal~in~\cite{KardasKP13}. In~\cite{DBLP:BordimCIN02,NakanoO00a} the authors considered energy-efficient algorithms for assigning unique identifiers to all stations. Energy efficient broadcast protocols in the related
 model of multi-hop radio networks was studied e.g., in~\cite{BroadcastEE1}. 


To the best of our knowledge, adversarial packet injections 
on multiple-access channel
were considered for the first time by Bender~\etal~in~\cite{BENDER05f} and  Chlebus~\etal~\cite{kow2}. The authors of the former
paper considered maximal possible throughput of randomized backoff  protocols in queue-free model, while in the latter work
deterministic distributed broadcast algorithms 
in the model of stations with queues were studied.
De Marco~\etal~\cite{DEMARCO20171} analyzed asynchronous channel.
Randomized counterparts of this problem can be found in the earlier paper~\cite{ChlebusKluwer}.
Backoff exponential algorithm improvements in scalability and throughput were discussed by Bender~\etal~in \cite{Bender:2016:SEB:2884435.2884482}.
Further results in this line considering the maximum rate for which stability of queues is achievable include Chlebus~\etal~\cite{kow} and Anantharamu~\etal~in~\cite{Lak8},
\cite{10.1007/978-3-642-22212-2_9},
wherein the authors considered a wide spectrum of models with respect to adversary's limitations and capabilities of stations and the channel (e.g., distinguishing collisions from the silence on the channel). 
In \cite{unlimit}, Bieńkowski~\etal~introduced the model with unlimited adversary that can inject packets into arbitrary stations with no constraints on their number nor rates of injection,
and pursued competitive analysis w.r.t. the optimal solution.
In all aforementioned papers the number of stations was known in advance. In~\cite{Lak5},~Anantharamu  and Chlebus investigated a 
channel with unbounded number of stations attached to it.


Hardware-related challenges were studied by Ogierman~\etal~\cite{Ogierman2018}, with a focus on adversarial jamming limited by the energy budget in MAC protocols for the SINR model.
Physical layer effects on a model of a single hop fading channel were also studied by Fineman~\etal~\cite{Fineman:2016:CRF:2933057.2933091} with particular attention to the spectrum reuse enabled by fading.

\subsection{Organization of this paper} 
In Section~\ref{Sec:Model} we present a formal model of the channel, stations and the adversary. 
Section~\ref{Sec:Adaptive} is devoted to the strongest, i.e., adaptive, algorithms, wherein stations can adopt their behavior to the communication channel and add some information to the transmitted packages. We construct an algorithm that needs only a constant number of stations being switched on in each round, which guarantees stability for an adversary even for $\rho=1$. In Section~\ref{Sec:Full2} 
we discuss a weaker class of protocols, namely the full-sensing protocols. We construct an algorithm with a collision-detection mechanism that is stable for an adversary with any $\rho<1$ . The weakest type of algorithms (acknowledgment based), wherein all actions are set before the execution of the algorithm, are discussed in Section~\ref{Sec:ACK}. In Section~\ref{Sec:Sim2} 
we present various experimental results for the constructed algorithms as well as a comprehensive comparison with commonly used (also in real-life systems) back-off protocols. 
Omitted details could be found in the Appendix, as well as conclusions and discussions in Section~\ref{Sec:Conc}.

\section{Model}\label{Sec:Model}  


%
We follow the classical model of a shared channel, c.f., \cite{Chlebus:2019:EEA:3323165.3323190,kow,hastad}, while also enriching it by restriction on the number of simultaneous channel activities (we call it a restrain and denote by $k$).
There are $n$  stations with unique names (also called IDs) from set $N=\{1,\ldots ,n\}$, attached to a shared transmission medium that make it a synchronous {\em  multiple-access channel}, or {\em shared channel}. The main properties of the channel are:

$\bullet$  Time is divided into slots of equal size, called rounds.\footnote{%
We restrict our attention to the synchronous ``slotted'' model, in which the stations use local clocks ticking at the same rate and indicate the same round numbers. We motivate it by 
portability of developed protocol -- 
asynchronized 
time windows can be transformed into
synchronous slots of equal size, under clock synchronization restriction.
}
Global round numbering is available to the stations.
Each round consist of phases: transmission, listening and data processing. The stations, according to their programs, attempt either to transmit in the first phase or to listen to the channel in the second phase.

$\bullet$  A packet is successfully received if its transmission does not overlap with any other transmission.
A packet successfully transmitted by a station is heard on the channel by all the stations in the listening stations
and is acknowledged by the transmitter.

\paragraph*{\bf Channel restrain} 

Following the model from  \cite{Chlebus:2019:EEA:3323165.3323190}, we introduce station operation modes: each station can be at one of two states -- \textit{switched on} (on-mode)  or \textit{switched off} (off-mode).
Only a switched-on station in a given round can transmit a packet or listen to the channel. 
In a round in which a station is switched on, the station can set its \textit{timer} to any positive integer~$c$, which results in the station spending the next $c$ rounds in the off-mode and returning to the on-mode immediately afterwards.
We assume that the adversary can inject packets into the station message queue independently from the station mode.

We say that the channel is {\it k-restrained} if at most $k$ stations can be in the on-mode in any given round.

\paragraph*{\bf Packets arrival} 

We assume that injected packets are kept in individual queues by each station, till they are successfully transmitted. We model the packet arrival by an adversary, who injects packets into the system queues. Different protocols can be compared under the same adversarial strategy. 
In our paper we consider the $(\rho,b)$-\textit{leaky-bucket} adversary \cite{kow}. It constrains the adversary according to two parameters: the injection rate $\rho$, being the maximum average number of packets injected per round, and the burstiness $b$  - the number of packets that can be injected simultaneously in the same round.\footnote{%
	There is also a second 
	model of the adversary, the so called \textit{window adversary}, which is seemingly weaker, c.f., \cite{ros} for details.}

\paragraph*{\bf Protocol families} 

Following the current classification in the literature, c.f., \cite{10.1007/978-3-642-22212-2_9,kow,kow2}, we consider the following classes of 
distributed protocols with respect to stations' capabilities:

\textit{Adaptive protocols} --- each station may access the history of transmissions and  each packet contains the unique  ID of the sender. Moreover the sender can add a constant number of bits to each packet, which other stations can read and make future decisions based on this information.

\textit{Full-sensing protocols} --- each station may access the history of transmissions and  each packet contains the unique  ID of the sender. However no extra bits to a packet can be added. 

\textit{Acknowledgment based protocols} --- each station runs a function of station ID and round, determining if transmission attempt should be made. 

Protocol classes described above reflect practical limitations of station's calculation and memory resources and its ability to analyze  channel state. \textit{Backoff} protocol can be described as randomized acknowledgment-based algorithm with fixed initial sequence of probabilities.

\paragraph*{\bf Protocol quality measures} 
We focus on the two following protocol quality and performance measures of an algorithm:

{\em Stability} -- if queues of all stations stay bounded by some function on model parameters ($n$, $\rho$, $b$), even in an infinite 
execution;

{\em Channel restrain} $k$ -- upper bound on the number of online stations in one round (also called $k$-restrained channel); 

{\em Throughput} -- the maximum injection rate $\rho$ for which all executions of the algorithm are~stable.


\section{Adaptive protocol}\label{Sec:Adaptive}

\subsection{12-O'clock adaptive protocol}

The \algoname{12-O'clock}$(n)$ algorithm, 
where $n$ is the number of stations in the system,
schedules exactly two stations to be switched on in a single round --- one in the {\it transmitting} role and another in the {\it listening} role. Since only one of those stations has the right to transmit, collision never occurs and 
 the channel restrain is~$2$.
%
The algorithm allows for any adversary burstiness value.

\textbf{High level description.} 
We call a group of $n$ consecutive rounds a {\it cycle} if the last round $r$ of the group satisfies $r = 0 \mod n$.
End-of-cycle (or 12-O'clock) rounds play an important role in coordination and decision making during the execution; they also motivate 
the name of the algorithm.

Every  station  keeps  an  ordered  list  of  all the  stations. These lists are the same in every station at the beginning of a cycle; at such a moment they represent one list, which we call the list. Initially, the list consists of all the stations ordered by their names.

Stations take the transmitting role in their order on the list. The  process  of  assigning  transmitting stations  to  rounds  can  be visualized as passing a virtual token from station to station, such that a station holding the token is in the transmitting role. Station spends one round in the listening role before taking the token, in order to learn the status of the channel.
When a cycle ends then the token is typically passed on to the next station on the list. The order determined by the list is understood in a cyclic sense, in that the first station assumes the transmitting role after the last one in the list has concluded its assignment.  An  exception  for  this  process  occurs  when  transmitting station  is  moved to become  the  head  of  the  list while keeping the token.

The exception is handled as follows:
a transmitting station $B$ holding  the token has the right to keep it when it has at least $3n$ packets in its queue. In such a case the station considers itself \state{Big} and informs other stations about its status, by suitably setting a toggle bit in messages. All of the stations while in the listening role, learn from this bit that they have no right to take the token.

Station $B$ has the right to keep the token until the first end-of-cycle (12-O) round with queue size not greater than $3n$ --- once this condition is fulfilled, the station considers itself to be \state{Last-Big}, has the right to hold the token for one more full cycle and informs other stations by suitably setting another toggle bit in its messages. 
By the end of this last cycle, all of the stations move $B$ to the head of the list. Starting with the next cycle stations follow their routine, with station $B$ being the head of the least and $B$ holding the token to transmit in the first round of the cycle.
This mechanism allows transmitting stations to stretch cycles, possibly indefinitely, should the adversary inject packets in a certain way, e.g., into one station only.

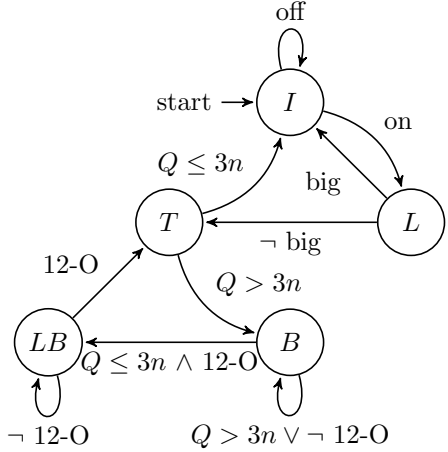
\begin{SCfigure}[1]
\begin{tikzpicture}[->,>=stealth',shorten >=1pt,auto,node distance=2.25cm,
                    semithick]

  \node[initial,state,minimum size=25pt] (A)                    {$I$};
  \node[state,minimum size=25pt]         (B) [below right of=A] {$L$};
  \node[state,minimum size=25pt]         (D) [below left of=A] 
  {$T$};
  \node[state,minimum size=25pt]         (C) [below right of=D] {$B$};
  \node[state,minimum size=25pt]         (E) [below left of=D]       {$LB$};

  \path (A) edge  [bend left] node {on} (B)
        (A) edge [loop above] node {off} (A)
        (B) edge              node {big} (A)
        (B) edge             node 
        {$\neg$ big} (D)
        (D) edge [bend right] node {$Q\leq3n$} (A)
        (D) edge [bend right] node {$Q>3n$} (C)
        (C) edge [loop below] node {$Q>3n \lor \neg$ 12-O} (C)
        (C) edge  node {$Q\leq3n$ $\land$ 12-O} (E)
        (E) edge   node {12-O} (D)
        (E) edge [loop below] node {$\neg$ 12-O} (E);
\end{tikzpicture}
\caption{Finite state machine for a station in 12-O'clock adaptive algorithm. {\it I} stays for Idle state; $L$ for Listening state; $T$ for Transmitting state; {\it B} for Big state; {\it LB} for Last-Big state. Note that from the system perspective there are two stations with starting state being Transmitting and Listening respectively, otherwise the starting state is Idle as shown. $Q$ stays for station queue size; $n$ for system size; {\it 12-O} for the end-of-cycle round; channel state heard by a station is described by: {\it big} for packets with big or last-big bit being toggled. Arrows {\it on} and {\it off} represent a global clock-based decision for choosing on-mode and off-mode respectively. All of the arrows happen with a tick of clock.} 
\label{12Oa:fsm}
\end{SCfigure}

\textbf{Technical description.} 
Operational station can be at one of the five states: \state{Idle}, \state{Listening}, \state{Transmitting}, \state{Big} or \state{Last-Big}. 
The last three states are given the right to transmit; they could be encoded by two bits when attached to the message by the transmitting station.
The \state{Listening} state is dedicated to listening, while in the \state{Idle} state the station is switched off.
We describe these states one-by-one later in this section. Finite state machine for the relationship between those states can be seen on {\it Figure \ref{12Oa:fsm}}.

\textbf{Initialization.}
In the beginning all but the first two stations are in the \state{Idle} state, while the one with the smallest ID is in \state{Transmitting} state and its successor is in the \state{Listening} state.

\textbf{Idle state.}
In this state the station does not access the channel, it only keeps updating its idling time
until the next wake-up --- each round decreases by $1$. 
The starting number of idling rounds is either $n$ or $n-1$ or $n-2$, depending on the 
state from which the station switches to \state{Idle} and the message on the channel,
see the description of \state{Listening} and \state{Transmitting} states below.
When the idling time decreases to zero, the state
switches to \state{Listening}.

\textbf{Listening state.} Station in the \state{Listening} state updates the local station list when the \state{Last-Big} transmission occurs on a channel. It changes its state to \state{Transmitting} upon receiving a message from a
station in the \state{Transmitting} state or upon no message received. Otherwise,
it becomes idle for the next $n-1$ or $n$ rounds until wake-up. The latter idling time
is caused by move of \state{Last-Big} station from behind of the \state{Listening} station location on the list of stations, to the front, therefore increasing the \state{Listening} station location by $1$.

\textbf{Transmitting state.}  
The \state{Transmitting} state is taken (from \state{Listening} state) by a station once per cycle in the round corresponding
to its current position on the list of stations, unless there is a \state{Big} or \state{Last-Big} station
in this round. 
Station in the \state{Transmitting} state changes its state to \state{Big} and transmits 
if its queue size is bigger than $3n$.
Otherwise, it transmits being in the \state{Transmitting} state, provided it has a packet in its queue, and changes its state to \state{Idle} (in order to awake in its listening turn during the next cycle, after $n-2$ rounds). 

\textbf{Big state.} 
At the end of each cycle, each \state{Big} station checks whether its queue size is still
bigger than $3n$; if not, it changes its state to \state{Last-Big}.
In any prior round, the \state{Big} station transmits a packet and remains in the same state.
The following property can be easily deducted: once a station changes its state to \state{Big} (which happens when being
in its regular \state{Transmitting} state), it stays there till the end of the cycle;
it may then continue throughout whole next cycles, until it changes to \state{Last-Big} state
at the end of one of them.

\textbf{Last-Big state.} Station in the \state{Last-Big} state transmits until the end of the cycle. It changes its state to the \state{Transmitting} in the end of the cycle after the last transmission happens. 
Note that, due to the condition of switching from \state{Big} to \state{Last-Big} state,
a station remains in the \state{Last-Big} state during the whole one cycle, from the beginning when it 
switched from the state \state{Big} to the end when it switches to the \state{Transmitting}.

\subsection{Analysis and bounds}

Consider the total size of the queues in the beginning of the cycle. If it is greater then $\ell = n(3n-1)+1$ we say that 
belongs to a {\em dense} interval, otherwise it belongs to a {\em sparse} interval (here we consider intervals of time). This way the execution of algorithm consists of interleaved dense and sparse intervals, each  
containing a number of whole cycles.

In relation to a fixed interval, we consider the following terminology: station is {\it pre-big} in a given round
of the interval if it has not been in the \state{Big} state during this interval before that round, and it is {\it post-big} if it has been at least once in the \state{Last-Big} state during the interval by that round. Station is {\it potentially-big} if its queue size is bigger than $3n$ (i.e., the size allows the station to become \state{Big} eventually) or it is in a \state{Big} or \state{Last-Big} state.
%
Observe that each station is pre-big in some prefix of the interval and post-big in some (disjoint) suffix of the interval; each of these periods could be empty or the whole interval.
In-between of being pre-big and post-big, a station is continuously in a \state{Big} state.

We define types of cycles 
depending on availability of 
\state{Big} and \state{Last-Big} stations:

\textbf{Type-1 cycle:} without any \state{Big} or \state{Last-Big} station. Token is being passed in the \algoname{Round-Robin} way,
    by adopting \state{Listening} and \state{Transmitting} states. This means that at any single round there is one station in the \state{Transmitting} state and one in the \state{Listening} state. 
    
\textbf{Type-2 cycle:} with a station $S$ starting to transmit as \state{Big} in some round of the cycle. Here, the token is being passed in the \algoname{Round-Robin} way by applying sequence of \state{Listening} and \state{Transmitting} states to each station on the list, until $S$ transmits. Since $S$ becomes \state{Big}, it keeps the token afterwards till the end of the cycle. Note that stations at \state{Big} and \state{Transmitting} states cannot occur simultaneously in the same round, because once there is a \state{Big} station all \state{Listening} stations immediately switch to \state{Idle} state instead of switching to \state{Transmitting} state. 
 
\textbf{Type-3 cycle:} with a \state{Big} station $S$ keeping the ``token to transmit'' for the whole cycle. All stations after waking-up in the \state{Listening} state will learn about the state of $S$ and become idle until their scheduled wake-up round in the next cycle. 

\textbf{Type-4 cycle:} with a \state{Last-Big} station $S$ keeping the token for the whole cycle. Station can be in the \state{Last-Big} state only for a one cycle and after being in the \state{Big} state (at the end of the previous cycle). All stations after switching from \state{Idle} to the \state{Listening} state will learn about the \state{Last-Big} state of $S$ and become idle until their scheduled wake-up round in the next cycle.

The local lists of stations stay synchronized in the beginning of cycles; in fact, only the type-4 cycle
changes the order of stations, and the whole cycle is needed to do it consistently in all stations (when they act as listeners) so that they all apply the move of the \state{Last-Big} station to the beginning of their local lists
by the end of the cycle.

\begin{lemma}
Each cycle is of one of the above four types.
\end{lemma}

\begin{proof}
Algorithm's initialization conditions enforce that the first cycle is of type-1 or type-2,
as there is no \state{Big} or \state{Last-Big} station in the beginning.
Type-1 can be followed only by the type-1 --- if there is no potentially-big station during the cycle, or by the type-2 cycle otherwise. 
In the type-2 cycle the \state{Big} station is chosen during the cycle, and thus the cycle can be followed by the type-3 cycle --- if the \state{Big} station queue size is above $3n$ at the end of the cycle, or by type-4 otherwise.
The case of type-3 cycles is the same as the ones of type-2 described above, as in both types
there is a \state{Big} station at the end (which determines conditiones for the next cycle); they can be followed
only by a cycle of type-3 or type-4.
The type-4 cycle can be followed by type-1 --- if there is no potentially-big station, 
or by type-2 cycle otherwise.
Using an iductive argument over cycles, it can be concluded that each
cycle is of one of the four defined types.
\end{proof}

\begin{lemma}
In any dense interval, a station can cause a silent round (i.e., is in state \state{Transmitting} but has an empty queue) at most $n-1$ times while being pre-big.
\label{lemma:pre-big}
\end{lemma}
\begin{proof}
Silent rounds occur when some station has a ``token to transmit'' but has no packets in its queue. 
Note that it is only possible for stations in \state{Transmitting} state, as stations in any of \state{Big} states have
more than $n$ packets in their queues.

Assume that station $S$ has no packets in its queue.
Within a dense interval, in each 
round
there is a potentially-big station. For any cycle, if potentially-big station is before $S$ in the list, then $S$ would receive no ``token'' or receive it and decrease its position in the list. The position of $S$ cannot decrease more then $n-1$ times, as there can be no potentially-big station after $S$ if it is last in the list. When $S$ is the last on the list it either never has a possibility to transmit or becomes potentially-big. Pre-big station life-cycle terminates once station is in the \state{Big} state by definition.
\end{proof}

\begin{lemma}
In any dense interval, post-big or in a \state{Big} state station causes no silent round. 
\label{lemma:post-big}
\end{lemma}

\begin{proof}
By definition of \state{Big} state, a station must have had more than $3n$ packets in its queue
in the beginning of the current cycle or in the round of the cycle when it turned into the \state{Big} state. Therefore, in each round of the cycle it has packets and causes no silent round.

A post-big station $S$ could be in a \state{Last-Big} state, \state{Big} state, \state{Transmitting} state or in one of the
other two states. In the latter case, it does not attempt to transmit, thus it cannot cause a silent round.
The case of \state{Big} state was already analyzed. If the station enters \state{Last-Big} state, it switches to this state
from the \state{Big} state 
having more than $2n$ packets in its queue,
thus in each round of the cycle it has packets and causes no silent round.

It remains to analyze the case when $S$ is in the \state{Transmitting} state.
Upon leaving the \state{Last-Big} state for the last time, it had at least $n$ packets in its queues and was
placed in the beginning of the list of stations, by the algorithm construction. 
Then, observe that $S$ has had an opportunity to transmit only at some type-2 cycle, when there is no potentially-big station before it on the list or when $S$ is potentially-big at the time it
switches from \state{Listening} to \state{Transmitting} state. In the latter case, 
instead of staying in \state{Transmitting} state it immediately switches to \state{Big} state,
which case we already analyzed in the beginning of the proof.
%
Otherwise (i.e., in the former case), either potentially-big station after $S$ becomes \state{Big}, 
which implies that in some
of the next cycles, it switches to \state{Last-Big} state and the position of $S$ on the list decreases without causing any more
silent rounds, or $S$ receives no ``token to transmit'' and so it cannot cause a silent round by default. 
The position cannot decrease more than $n-1$ times, because there can be no potentially-big station after $S$ if $S$ is the last on the list (the argument is similar to the one from the proof of Lemma~\ref{lemma:pre-big}). 
Since $S$ had at least $n$ packets when switching from its \state{Last-Big} state,
it can transmit and decrease its position at most $n-1$ times or become \state{Big}, whatever comes first;
in any case, it has at least one packet when transmitting.
\end{proof}

\begin{theorem}
\label{theorem:adaptive}
The 12 O'clock adaptive protocol 
 achieves throughput $1$ on a channel with restrain $2$ and the maximum number of packets stored in a round
is at most $O(n^2+b)$.
\end{theorem}

\begin{proof} 
Consider an adversarial injection pattern with rate $\rho=1$ and a burstiness $b$.
Within a sparse interval, there can be no more then $\ell + n + b$ packets in the stations
at the end of any cycle for dense interval threshold $\ell$.
Indeed, the biggest possible number of packets that the system can start a cycle with is equal to $\ell$, 
and the adversary can inject no more then $n+b$ packets in $n$ consequent rounds of the cycle. 
Once the queue size becomes greater than $\ell$ in the beginning of a cycle, 
the sparse interval terminates and the dense interval begins.

In the remainder, we focus on dense intervals.
Note that in the beginning of a dense interval, the number of packets in the system
is at most $\ell+n$ plus the burstiness above the injection rate (upper bounded by $b$);
indeed, as in the beginning of the preceding cycle the interval was sparse, the number of packets
was not bigger than $\ell$, and during that cycle the adversary could inject at most
$n$ packets accounted to the injection rate plus the burstiness.

Within any dense interval, a station in the \state{Big} or \state{Last-Big} state is guaranteed to be in each cycle, by the pigeon-hole principle. It makes type-1 cycle impossible to occur. 
Consider type-3 and type-4 cycles: during those cycles packet is transmitted in every round, and 
thus a silent round cannot occur; hence the number of packets does not grow (except
of burstiness above the injection rate, but this is upper bounded by $b$ at any round of the interval,
by the specification of the adversary).
In type-2 cycles, post-big stations cannot cause silent round, by Lemma \ref{lemma:post-big},
and stations in \state{Big} state cannot cause silent rounds as they always have more than $2n$ packets
pending.
Hence, type-2 cycles may have silent rounds caused only by pre-big stations. However, there can be no more than $n-1$ pre-big stations in the system in the beginning of the dense interval (because there is at least one potentially-big station). Each pre-big station can cause no more than $n-1$ silent rounds, by Lemma \ref{lemma:pre-big}. Observe that in each cycle with a silent round some potentially-big station will change its state to \state{Big} --- silent round would not occur if there was potentially-big station with higher position in the list than any empty station. 
Hence, there can be no more than $n-1$ cycles with silent rounds caused by same (pre-big)
station. To summarize, there are at most $n-1$ cycles with silent rounds for each of at most $n-1$ pre-big stations, resulting in the upper bound of $\ell + (n-1)^2 + n+b$ on system queues.
Since only one of those stations has the right to transmit, collision never occurs and 
channel restrain is~$2$.
\end{proof}


\section{12 O'clock full-sensing protocol with collision detection}\label{Sec:Full2}
12-O'clock full-sensing protocol works similarly to its 
control-bits counterpart from Section~\ref{Sec:Adaptive}, however a decision to change state is based on information who transmitted packet, since adding bits is not allowed for full-sensing protocols. 
To overcome the lack of additional information bits, we implement  precise control mechanism to the ``out-of-order transmissions'' and the collisions enforced by them so that \state{Big} stations could
	be identified. Combined with recognition of ID attached to successful transmissions, any station can learn about \state{Big} station and adjust to it, with some small waste of transmission
	and increased delay.

In our algorithm, typically stations in the \state{Listening} state discover Big stations by reading the ID of station transmitting on the channel and comparing it to the predecessor ID from the list --- if they do not match then the listening station(s) deduct that the transmitter
is in \state{Big} state.
By distinguishing silence from collision, the algorithm is able to manage borderline cases. 
However, due to collisions, the channel restrain is $3$ and the protocol 
achieves slightly smaller throughput $1-\frac{1}{n}$ (recall however that no full-sensing protocol could reach throughput $1$
even on non-restrained channels).

\textbf{Key differences} from the adaptive version of the algorithm are: state of the channel (i.e. silence, transmission or conflict) is used instead of the control bits attached to a packet, and there is no \state{Last-Big} state. As this class of algorithms restricts control-bits usage, we use station knowledge of the order of stations combined with the channel state to identify the transmitting station. Indeed, the round-robin cycle of algorithm execution does not require any additional synchronization. For \state{Big}-state the stations note that all but one stations in the system can learn this state from unique transmitter ID attached to a packet and heard on the channel, as it does not match the expected (by the order of the list) transmitter ID. Similarly, we use a conflict on the channel to indicate that there is a station in \state{Big} state present on the channel.
Because of those changes, the synchronization mechanism of the order of stations' list has been modified and \state{Last-Big} state becomes unnecessary. 

A high level diagram of the
finite state machine for it can be seen on {\it Figure \ref{12Of:fsm2}}.
Complete description and analysis of 12-O'clock full-sensing algorithm can be found in Appendix \ref{Sec:Full}.
%
\begin{SCfigure}[1]
\centering
\begin{tikzpicture}[->,>=stealth',shorten >=1pt,auto,node distance=2.4cm,
                    semithick]

  \node[initial,state,minimum size=25pt] (A)                    {$I$};
  \node[state,minimum size=25pt]         (B) [below right of=A] {$L$};
  \node[state,minimum size=25pt]         (D) [below left of=A]  {$T$};
  \node[state,minimum size=25pt]         (C) [below right of=D] {$B$};

  \path (A) edge [bend left=45]   node {on} (B)
        (A) edge [loop above] node {off} (A)
        (B) edge [loop right] node {$\lightning$} (B)
        (B) edge [bend right] node {$\neg$ pred} (A)
        (B) edge  node {silence $\lor$ pred} (D)
        (D) edge [bend right=15] node {$Q\leq3n$ $\lor$ $\lightning$} (A)
        (D) edge [bend right] node {$Q>3n$ $\land$ $\neg \lightning$} (C)
        (C) edge [bend left=45]  node {$Q\leq2n$ $\land$ 12-O} (D)
        (C) edge [loop below] node {$Q>2n \lor \neg$ 12-O} (C);
\end{tikzpicture}
\caption{
A high level diagram of the
finite state machine for a station in 12-O'clock full-sensing algorithm. 
 {\it I} stays for Idle state; $L$ for Listening state; $T$ for Transmitting state; {\it B} for Big state. Note that from the system perspective there are two stations with starting state being Transmitting and Listening respectively, otherwise the starting state is Idle as shown. $Q$ stays for station queue size; $n$ for system size; {\it 12-O} for the end-of-cycle round. Channel states are described as: {\it silence} for no transmission, $\lightning$ for collision, {\it pred} stays for a message transmitted by the predecessor of the station. Arrows {\it on} and {\it off} represent a global clock-based decision for choosing on-mode and off-mode respectively. All of the arrows happen with a tick of clock.} \label{12Of:fsm2}
\end{SCfigure}

\section{Acknowledgment-based protocols}\label{Sec:ACK}
 
 In this section we consider acknowledgment-based protocols in $k$-restrained model. First we prove two limitations for this class of protocols.
 Then we present an algorithm that is throughput-optimal up to a multiplicative polylogarithmic factor.
 Our algorithm is based on a new type of selectors, $k$-light selectors, for which we derive upper and lower bounds, as well as polynomial construction.
 
\subsection{Limitations of acknowledgment-based protocols with restrain $k$} 
 
\begin{lemma}
There is no correct,   acknowledgment-based algorithm in the $k$-Restrain model with channel restrain $k < n$ 
without global-clock mechanism for any $\rho > 0$.

\end{lemma}
\begin{proof}
 We say that protocol with channel restrain $k$ and injection rates $\rho$ is correct, when queues of all 
 stations stay bounded at all times independently from adversary strategy and for any round 
 the number of online stations is at most $k$.
 
Assume that $P$ is correct deterministic acknowledgment-based protocol, within $k$-Restrain model without global clock in the system of $n$ stations. 
Then for each station $S_i$, there is a default starting sequence $p_i$, where $i$ is the index of the station. Because $P$ is correct, each $p_i$ contains a first occurrence of transmission bit $1$. Let $t_i$ be the position of the first transmitting bit  in the sequence $p_i$. 
 Because system is not equipped in the global clock mechanism, stations' starting rounds are set by adversary. Let us  say that $s_i$ is a global timeline start moment of station $S_i$. It follows that first transmission of station $i$ occurs at round $s_i+t_i$. 
 In order to overload the system adversary follows the strategy: choose round $e$ as $e = max\{t_1,\ldots ,t_n\}$; start station $S_i$ at round $s_i = e - t_i$. Then all $n > k$ stations transmit at round $e$ and thus $P$ has to overflow channel restrain $k$.
 \end{proof}

\begin{theorem}
Any acknowledgment-based algorithm with global clock in the $k$-Restrain model with channel restrain  $k < n$ 
cannot achieve throughput higher than $\min \{ \frac{k}{n}, \frac{1}{3\log n }\}$.
\end{theorem}
\begin{proof}
To prove the theorem, assume first that $\frac{k}{n}\le \frac{1}{3\log n }$.
Consider a period of $\tau$ consequent rounds. 
Suppose, to a contradiction, that during  $\tau$ rounds  the adversary can inject $\tau\cdot k/n + 1$ packets. The channel restrain of $k$ implies that at most $k$ stations can be active and, therefore, during $\tau$ rounds there could be at most $\tau\cdot k$ activities in total.
There are $n$ stations in the system, hence, by pigeon-hole principle, there is a station allowed to transmit at most $\tau\cdot k/n$ packets during $\tau$ rounds.
Acknowledgement-based protocols with global clock provide adversary with a power to know stations schedules in advance, as the adversary can calculate values of the protocol function for any round and for each station; hence,
it can pick a station $S$, such that the number of scheduled switch-on rounds is minimal within the system. Once $S$ is chosen the adversary can inject $\tau\cdot k/n + 1$ packets into the queue of $S$. Queues of arbitrary length would be generated by iteration of the procedure, thus the system cannot be stable, which results in contradiction.
This proves that $\rho$ cannot exceed~$\frac{k}{n}$. 
The second case when the minimum formula equals to $\frac{1}{3\log n }$ follows directly from Thm.~5.1~in~\cite{kow}.
\end{proof}

\subsection{Algorithm}
In this section we present an algorithm  working in $k$-restrain channel and achieving throughput 
$\Theta(\frac{k}{n \log^2 (n)})$. 
We start with 
introduction and thorough study of $k$-light selectors.

\subsubsection{$k$-light Selectors}
Let us consider a set $N=\{1,\ldots ,n\}$ and its subsets  $S,X,Y \subset N$. We say that 
$S$ \textit{hits} $X$ if $|S \cap X| = 1$. We say that $S$ \textit{avoids} $Y$ if 
$|S \cap Y| = 0$.

\begin{definition}
We say that a family $\mathcal{S} \subset 2^{N}$ is a $(n,\omega)$-selector if 
for any subset  $X\subset N$ such that $\omega/2  \leq |X|\leq \omega$ there are $\omega/4$ elements 
 hit by subsets from $\mathcal{S}$. 
\end{definition}

Note that this definition is a special case of a \textit{selective family} introduced in \cite{Chlebus2005}. 
The intuition behind $\mathcal{S}$ is as follows: we can ``separate'' at least a fraction of elements 
of any subset $X$ (of appropriate size) using sets that belong to $\mathcal{S}$.

\begin{definition}
 We say that $\mathcal{S}=(n,\omega)$-selector is $k$-light if any $S\in \mathcal{S}$ satisfies $|S|\le k$.
\end{definition}

\begin{theorem}\label{thmSelectors}
There exists a $k$-light $(n,\omega)$-selector  of size 
$    m = O\left ((\omega + n/k) \log n\right ) $.
\end{theorem}

\begin{proof}
The first part of the formula, $O\left (\omega \log n\right ) $ for $\omega \ge \frac{n}{k}$, comes from generalizing the reasoning in \cite{chrobak}.
Let us assume that $\omega>1$ and $\omega | n$.  Let $m$ be the size of a selector to be fixed later. 
Let us choose  independently  $m$ random  subsets of $\{1,\ldots , n\}$ of size $l=\frac{n}{\omega}$. 
That is, $\mathcal{S} = (S_1, \ldots , S_m)$ is a random family.  Let us consider any \textbf{fixed} sets $X,Y \subset \{1,\ldots ,n\}$,
such that $\omega/4 \leq |X| \leq \omega$; $|Y| \leq \omega/4$ and a \textbf{random} $S_i$.
\begin{gather*}
\Pr[S_i \mbox{ avoids } Y \mbox{ and hits } X]= 
\frac{{ |X| \choose 1 } { n - |X| - |Y| \choose l-1 }}{ {n \choose l}}  = 
|X| \cdot l \cdot \frac{(n - |X| - |Y|)^{\underline{l-1}}}{n^{\underline{l}}} = \\
\frac{|X| \cdot l}{n-l +1}\prod\limits_{i=0}^{l-2}\frac{n-|X|-|Y|-i}{n-i} > 
\frac{\frac{\omega}{4} \cdot l}{n}\prod\limits_{i=0}^{l-2}\frac{n-\frac{5}{4}\omega-i}{n-i} \geq 
\frac{\frac{\omega}{4} \cdot l}{n} \left(1 - \frac{\frac{5}{4}\omega}{n-l+2} \right)^{l-1} \geq \\
\frac{\frac{\omega}{4} \cdot \frac{n}{\omega}}{n} \left(1 - \frac{\frac{5}{4}\omega}{n/4} \right)^{l-1} \geq 
\frac{1}{4} \left(1 - \frac{5}{n/\omega} \right)^{\frac{n}{\omega}-1} \geq \frac{1}{4} \exp(-5)= c > 0
\ .
\end{gather*}

Let us bound the  probability that  for \textbf{any} sets $X,Y$ such that $\omega/4 \leq |X| \leq \omega$ and $|Y| \leq \omega$ 
there exists an $i$ such that $S_i$ hits $X$ and avoids $Y$. The probability of complementary event can be roughly bounded  as follows:
\begin{gather*}
\sum_{|X| = \frac{\omega}{2}}^{\omega} \binom{n}{|X|} \sum_{y = 0}^{\omega} \binom{n}{|Y|} (1 - c)^m \leq \nonumber 
    \omega^2n^{2w}(1 - c)^m \leq \nonumber 
    n^{4\omega}(1 - c )^m \leq \nonumber 
    e^{4 \omega \ln n - m \ln (1-c)} < 1
    \ . 
\end{gather*}

Note that the last inequality holds for some  $m = O(\omega \log n )$. That is, for such $m$ the random structure 
$\mathcal{S}= (S_1,S_2,\ldots , S_n)$ with probability greater then zero hits \textbf{any} $X$ and avoids \textbf{any} $Y$ of an appropriate sizes.  
Thus such structure must exist and in consequence we can take $\mathcal{S}$ and use it for the  reminder of~the~proof. 

 Now we show that $\mathcal{S}$ is a $(n,\omega)$-selector. Let us take any $X$  
such that $ \omega/2 \leq |X| \leq \omega$ and $Y=\emptyset$. By the property of $\mathcal{S}$ there exists 
$S_{i_1}$ such that it hits $X$. Let $\{r_1\} =  |S_{i_1} \cap X|$. Now let us construct 
$X=X\setminus \{r_1\}$ and $Y = Y \cup \{r_1\}$. Since still $\omega/4 \leq |X|<\omega$ and $|Y| \leq \omega/4$ we can find 
$S_{i_2}\in \mathcal{S}$, such that it hits the truncated $X$ and avoids $Y=\{r_1\}$ thus there exists 
$r_2= |S_{i_2} \cap X|$. Then we set $X=X\setminus \{r_2\}$ and $Y = Y \cup \{r_2\}$.  We iterate such separation  $\omega/4$ times to get $\omega/4$ distinct elements 
that are chosen from the initial $X$. 
Thus we get the first case of the theorem. 

\medskip

To prove the second part of the formula, $O\left ((n/k) \log n\right ) $ for $\frac{n}{k}>\omega$, first we need to construct an 
$\frac{n}{\omega}$-light selector $\mathcal{S}'$
of size $m= O(\omega \log n)$. Clearly, this is possible using the above construction.
Then we need to partition  each $S_i \in \mathcal{S}'$ into  $\lceil\frac{n}{k \omega}\rceil$ sets of size 
at most $k$ to obtain a ``diluted'' selector. This results in $m = O(\frac{n}{k \omega} \omega \log n) = O(\frac{n}{k }  \log n)$ 
sets of size at most $k$. 
\end{proof}

\subsubsection{Construction of selector in polynomial time}

We present a polynomial time construction of $k$-light selectors, which is only slightly worse than existential result.
It uses two major components: dispersers and superimposed codes.

\noindent
{\bf Dispersers.}
A bipartite graph~$H=(V,W,E)$, with set $V$ of inputs
and set $W$ of outputs and set~$E$ of edges, 
is a \textit{$(n,\ell,d,\delta,\ep)$-disperser} 
if it has the following 
properties:
$|V|=n$ and $|W|=\ell d/\delta$;
%
each $v\in V$ has $d$ neighbors;
%
for each $A\subseteq V$ such that $|A|\ge \ell$,
the set of neighbors of $A$ is of size at least $(1-\ep)|W|$.
Ta-Shma, Umans and Zuckerman~\cite{TUZ} showed how to
construct, in time polynomial in~$n$, 
an $(n,\ell,d,\delta,\ep)$-disperser for any $\ell\le n$, some $\delta=O(\log^3 n)$
and $d=O(\text{polylog }n)$.

\noindent
{\bf Superimposed codes.}
A set of $b$ binary codewords of length $a$, 
represented as columns of an $a\times b$ binary array,
is a \textit{$d$-disjunct} {\em superimposed code}, if it
satisfies the following property:
no boolean sum of columns in any set~$D$ of $d$ columns can cover
a column not in~$D$.
Alternatively, if codewords are representing subsets of $[a]$, then $d$-disjunctness means
that no union of up to~$d$ sets in any family of sets~$D$ could cover a set
outside~$D$.
Kautz and Singleton~\cite{KS} proposed a $d$-disjunct superimposed codes for $a=O(d^2 \log^2 b)$,
which could be constructed in polynomial time. 

\noindent
{\bf Polynomial construction of light selectors.}
We show how to construct $k$-light $(n,\omega)$-selectors of length
$m = O\left (\omega \text{ polylog } n\right )$ for $k \geq \frac{n}{\omega}$  and  
$m = O\left (\frac{n}{k} \text{ polylog } n\right )$  for  $k < \frac{n}{\omega}$,
in time polynomial in~$n$.
This is equivalent to constructing $k$-light $(n,\omega)$-selectors of length
$m = O\left ((\omega+n/k) \text{ polylog } n\right )$ 
in time polynomial in~$n$.
The construction combines specific dispersers with superimposed codes.
%
Let $0<\ep<1/2$ be a constant.
Let $G=(V,W,E)$ be an $(n,\omega/4,d,\delta,\ep)$-disperser for some $\delta=O(\log^3 n)$
and $d=O(\text{polylog }n)$, constructed in time polynomial in $n$, c.f., 
\cite{TUZ}.
%
%
Let $\cM=\{M_1,\ldots,M_a\}$ be the rows of the $c\delta$-disjunct 
superimposed code array of $n$ columns, for $a=O((c\delta)^2 \log^2 n)$,
constructed in time polynomial in $n$, c.f., Kautz and Singleton~\cite{KS};
here $\delta$ is the parameter from the disperser~$G$ and $c>0$ is a sufficiently large constant.
W.l.o.g. we could uniquely identify an $i$th of the $n$ columns of the superimposed code with
a corresponding $i$th node in~$V$.

For a constant integer $c$ we define a 
$k$-light $(n,\omega)$-selector $\cS(n,\omega,k,c)$ of length
$m\le \min\{n,a|W|\alpha\}$, for some $\alpha$ to be defined later, which consists of sets $S_i$, 
for $1\le i\le m$.
Consider two cases. In the case of $n\le m|W|\alpha$, we define $S_i=\{i\}$.
In the case of $n> a|W|\alpha$, we first define sets $F_j$ as follows:
for $j=xa+y\le a|W|$, where $x,y$ are non-negative integers 
satisfying $x+y>0$, $F_j$ contains all the nodes $v\in V$ 
such that $v$ is a neighbor of the $x$-th node in~$W$ and $v\in M_y$;
i.e., $F_{x\cdot a+y} = M_y\cap N_G(x)$.
Next, we split every $F_j$ into $\lceil |F_j|/k \rceil$ subsets $S$ of size at most $k$ each,
and add them as elements of the selector $\cS(n,\omega,k,c)$.
Note that each set $S_i$ from $\cS(n,\omega,k,c)$ corresponds to some
set $F_j$ from which it resulted by the splitting operation; we say that $F_j$ is a parent of $S_i$ and
$S_i$ is a child of $F_j$.
In this view, parameter $\alpha$ in the upper bound $m\le a|W|\alpha$ could be interpreted as
an amortized number of children of a set $F_j$. We will show in the proof of the following theorem that
$\alpha \le \frac{nd \cdot (c\delta)^2 \log^2 n}{k} \cdot \frac{1}{a|W|} + 1$.
The proof of the following theorem is deferred to Appendix~\ref{sec:construction-proof}.

\begin{theorem}
\label{thm:SelectorsConstructive}
$\cS(n,\omega,k,c)$ is a $k$-light $(n,\omega)$-selector  of length 
$m = O\left (\min\{n,(\omega+n/k) \text{ polylog } n\}\right )$ for a sufficiently large constant $c$, and is constructed in time polynomial in $n$.
\end{theorem}

\subsubsection{k-light Interleaved Selectors protocol}


Let us assume that $n$ is a power of $2$ and thus $\log n $ is an integer. We consider a sequence of $\mathcal{S}_1,\ldots ,\mathcal{S}_{\log (n)} $, where $\mathcal{S}_i$ is  $k$-light $(n,2^{i})$-selector of size $m_i$. Moreover, let $S^{j}_{i}$ be the $j$-th set of the $i$-th selector. That is, $\mathcal{S}_i = \{S^{1}_{i},\ldots , S^{m_i}_{i} \}$.
Let us consider the round number $t$ that can be uniquely represented 
as $t= j\log n + i$ for $1\leq i \leq \log n $ and $j \geq 0$.
Station $x$  transmits in the $t$ round if and only if 
 $x$ has a packet to be transmitted  and 
 $x \in S^{j \mod m_i  + 1}_i$.
The order  sets of selectors ``activating'' stations is crucial for performance of the algorithm and motivate its name. 
This order is depicted on the Figure~\ref{orderACK}.

\begin{figure}
\centering
  \begin{minipage}[c]{0.3\textwidth}
\centering
   \begin{tikzpicture}[node distance=0.5cm]
 \tikzset{every node}=[font=\small]
		\node (001) at (0, 1.8) {$S_1^1$};
		\node (002) at (0.5, 1.8) {$S_1^2$};
		\node (003) at (1, 1.8) {$S_1^1$};
		\node (004) at (1.5, 1.8) {$S_1^2$};
		\node (005) at (2, 1.8) {$S_1^1$};
		\node (006) at (2.5, 1.8) {$S_1^2$};
		\node (007) at (3, 1.8) {$S_1^1$};
		\node (008) at (3.5, 1.8) {$S_1^2$};
		\node (009) at (4, 1.8) {$S_1^1$};
		\node (010) at (4.5, 1.8) {$S_1^2$};
		
		\node (101) at (0, 1.2) {$S_2^1$};
		\node (102) at (0.5, 1.2) {$S_2^2$};
		\node (103) at (1, 1.2) {$S_2^3$};
		\node (104) at (1.5, 1.2) {$S_2^4$};
		\node (105) at (2, 1.2) {$S_2^5$};
		\node (106) at (2.5, 1.2) {$S_2^1$};
		\node (107) at (3, 1.2) {$S_2^2$};
		\node (108) at (3.5, 1.2) {$S_2^3$};
		\node (109) at (4, 1.2) {$S_2^4$};
		\node (110) at (4.5, 1.2) {$S_2^5$};
		
		\node (201) at (0, 0.6) {$S_3^1$};
		\node (202) at (0.5, 0.6) {$S_3^2$};
		\node (203) at (1, 0.6) {$S_3^3$};
		\node (204) at (1.5, 0.6) {$S_3^4$};
		\node (205) at (2, 0.6) {$S_3^5$};
		\node (206) at (2.5, 0.6) {$S_3^6$};
		\node (207) at (3, 0.6) {$S_3^1$};
		\node (208) at (3.5, 0.6) {$S_3^2$};
		\node (209) at (4, 0.6) {$S_3^3$};
		\node (210) at (4.5, 0.6) {$S_3^4$};
		
		\draw[->]
		(001) edge (101)
		(101) edge (201)
		(201) edge (002)
		
		(002) edge (102)
		(102) edge (202)
		(202) edge (003)
		
		(003) edge (103)
		(103) edge (203)
		(203) edge (004)
		
		(004) edge (104)
		(104) edge (204)
		(204) edge (005)
		
		(005) edge (105)
		(105) edge (205)
		(205) edge (006)
		
		(006) edge (106)
		(106) edge (206)
		(206) edge (007)
		
		(007) edge (107)
		(107) edge (207)
		(207) edge (008)
		
		(008) edge (108)
		(108) edge (208)
		(208) edge (009)
		
		(009) edge (109)
		(109) edge (209)
		(209) edge (010)
		
		(010) edge (110)
		(110) edge (210);
\end{tikzpicture}
  \end{minipage}\hfill
  \begin{minipage}[c]{0.6\textwidth}
\centering
\caption{Interleaved Selectors: $\mathcal{A} = \{\mathcal{S}_1,\mathcal{S}_2,\mathcal{S}_{3} \}$, where $\mathcal{S}_1 = \{S^{1}_{1},S^{2}_{1} \}$, $\mathcal{S}_2 = \{S^{1}_{2},\ldots , S^{5}_{2} \}$ and $\mathcal{S}_3 = \{S^{1}_{3},\ldots , S^{6}_{3} \}$.} 
\label{orderACK}
  \end{minipage}
\end{figure}
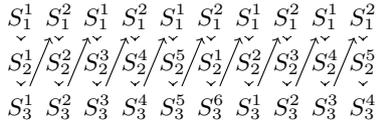

\subsection{Protocol analysis}

Obviously in a single round at most $k$ stations can transmit, since the sets $S_{i}^{j}$ consist of at most $k$ elements.
%
Let us now investigate the performance of the protocol. 

\begin{theorem} 
\label{t:ack-based}
 Let us assume that in round $t$ there are $r$ stations with nonempty queues, such that $2^i \leq r < 2^{i+1}$.
 The system will transmit at least $2^{j} /16$ packets before the round $t'=t + 8\sum_{l=i}^{j} m_{l} \log n $ for some $j\geq i$. 
\end{theorem}

\begin{proof}
Let us first consider a set $X_{0} \subset \{1,\ldots , n\}$ of stations such that $|X_{0}|=r$ and $2^{i-1}\leq r < 2^i$. 
Let   $\mathcal{S}_i = \{S^{1}_{i},\ldots , S^{m_i}_{i} \}$
be a $(n,2^i)$-selector and  $\mathcal{S}_{i+1} = \{S^{1}_{i},\ldots , S^{m_{i+1}}_{i} \}$ be a $(n,2^{i+1})$-selector for some $i<\log (n)$.
We assume that stations from $X_{0}$ have  nonempty queues of messages. We observe all stations  during  $T=m_i+m_{i+1}$ rounds.  
We assume that the adversary can add packets to queues (even to initially empty queues) during the execution of the algorithm. 
Let $X_{t}$ be the set of nonempty stations in the round $t$. In the $j$-th round  
stations from $X_{j}\cap S^{j}_{i}$ transmit for $j<m_i$ and $X_{j}\cap S^{j-m_i}_{i+1}$ for $j\geq m_i$. 
In other words, in consecutive rounds  transmit nonempty stations pointed by sets from  $\mathcal{S}_{i}$, 
then stations from~$\mathcal{S}_{i+1}$.

\begin{lemma}\label{lemP}
 If less then  $2^{i}/16$ different stations has transmitted during $T$ rounds of the process  
 then $|X_{T}| \geq \min \{r+ 2^i/8, 2^{i+1} \}$.
\end{lemma}

\begin{proof}
 Let $Y=\bigcup_{i=1}^{T} X_{i}\setminus X_{0}$ be the set of all stations filled by the adversary during the process. Let $\mathcal{O}^{*}$ be the set of stations that transmitted during the process. 
Moreover, let $\mathcal{T}(X)$ denote the set of the stations that transmitted at least once in the \textit{static} case with the initial set $X$ of nonempty stations, i.e. when the adversary does not add any messages. 

Clearly, $|\mathcal{T}(X_{0}\cup Y)| \leq  |\mathcal{O}^{*}| + |Y|$. Indeed, adding  $Y$ to the set of nonempty stations can increase the number of transmitting stations only by $|Y|$. On the other hand if a transmission of a station is blocked in the original process 
it must be also blocked in the case if all $X_{0}\cup Y$ stations are nonempty at the beginning.

Let us consider two cases. In the first we assume $|X_{0}\cup Y| < 2^ {i+2}$. In follows that  $|\mathcal{T}(X_{0}\cup Y)|\geq 2^i / 4$ because of the properties of  selectors. 
Thus $2^i / 4 \leq  |\mathcal{O}^{*}| + |Y|$. We assumed however that   $|\mathcal{O}^{*}| < 2^{i}/16$ , thus $|Y| > 3/16 \cdot 2^i $. That is, the adversary added messages to at least $3/16 \cdot 2^i$ initially nonempty stations but less then  $2^{i}/16$ has transmitted. Finally in he round $T$ a least  $r  + 2^{i}/8$ are nonempty. 

In the remaining case, if $|X_{0}\cup Y| >  2^ {i+2}$ and  only at most stations $2^{i}/16$ transmitted, the lemma holds trivially. 
\end{proof}

Note that in \textbf{any} contiguous segment of $(m_i + m_{i+1}) \log n$ rounds all sets of stations  with nonempty queues 
from  $\mathcal{S}_{i},\mathcal{S}_{i+1}$ are allowed to transmit  (see Fig~\ref{orderACK}). 
Following Lemma~\ref{orderACK} after $(m_i + m_{i+1}) \log n$ executed  rounds  
at least one of the three events 
occurred: (1)  $2^{i}/16$ transmitted; (2) the number of stations with nonempty queues 
increased by  $2^{i}/8$; (3) there is at least $2^{i+1}$ nonempty queues.

Note that event (3) may occur at most $\log n  - i$ times, similarly event (2) may occur at most  $8(\log n  - i)$ 
times till reaching the state of at least $2^{n-1}$ nonempty stations. Thus, after at most 
$\sum^{\log n-1}_{i=1} (m_i + m_{i+1}) \log n + m_{\log n} \log n  = O(\frac{n}{k}\log^2 n)$ rounds at least a fraction of nonempty stations will transmit
at least one packet. 
\end{proof}

Combining Theorem~\ref{t:ack-based} with Theorem~\ref{thmSelectors} we get:

\begin{corollary}
The protocol 
achieves throughput $\Theta( \frac{k}{n \log^2n})$ on $k$-restrained channels.
\end{corollary}

\section{Algorithms simulations}\label{Sec:Sim2}
In order to evaluate efficiency of developed protocols, we performed simulations for both new and existing algorithms and compared the results. We analyzed the impact of the execution length, system size and injection rates on the queue sizes and throughput/queue-size efficient channel restrain.
Experiments were limited to a setting strictly based upon the model. Complete simulation assumptions, settings and results can be found in Appendix \ref{Sec:Sim}.

\begin{figure*}[t!]
    \begin{subfigure}[t]{0.47\textwidth}
	\includegraphics[height=1.6in]{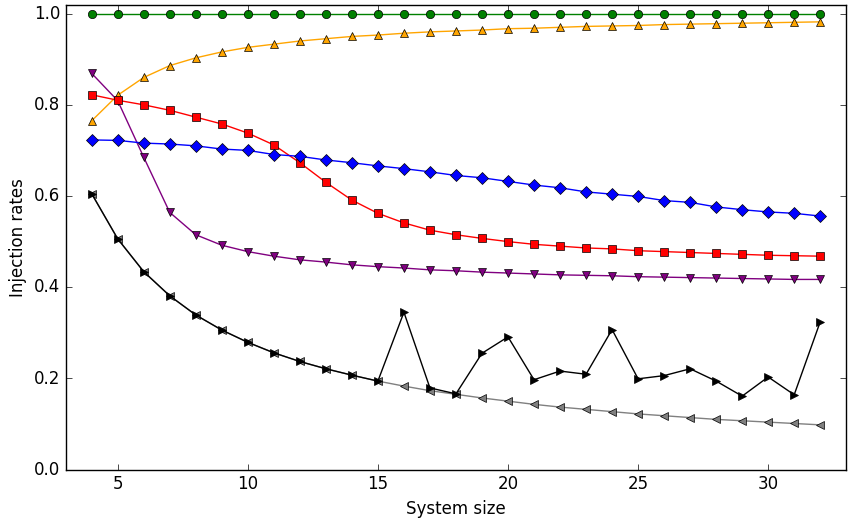}
	\caption{Stable injection rates over queue size for system size $n \in \{4,5,\ldots,32\}$.} 
	\label{stable2}
    \end{subfigure}
    ~ 
    \begin{subfigure}[t]{0.53\textwidth}
\includegraphics[height=1.63in]{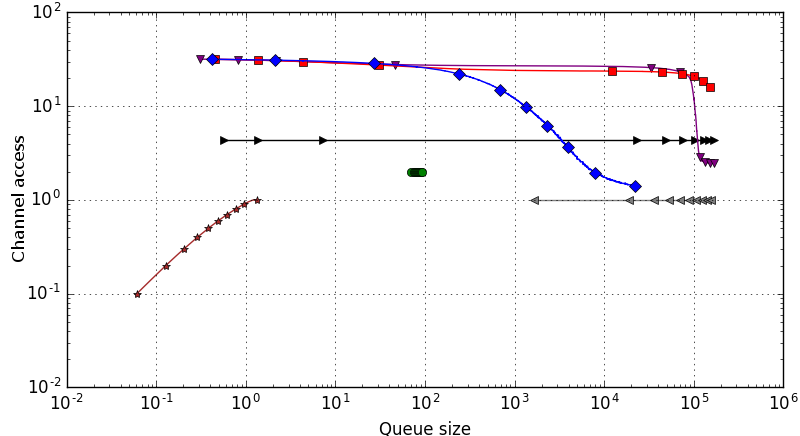}
\caption{Average round channel accesses against  
	queue sizes in logarithmic scale,  with markers set every $0.1$ for injection rates $\rho \in \left[0,1 \right]$ and system size $n=32$.}
  \label{EQ2}
    \end{subfigure}
    ~
    \begin{subfigure}[t]{1 \textwidth}

	\includegraphics[height=0.25in]{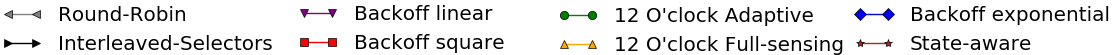}
	\caption{Relation between markers, colors and protocols.}
	\label{LP2}
    \end{subfigure}
    \caption{Average round channel access and stable injection rates by system size.}
\end{figure*}

We collate Adaptive and Full-sensing versions of the \algoname{12-O'clock} algorithm with Backoff exponential and polynomial algorithms. Additionally, we take into account acknowledgment-based algorithms: \algoname{Round-Robin} and $8$-light \algoname{Interleaved-Selectors}.

\paragraph*{Selected results.}

We studied throughput for the system sizes $n\in\{4,5,\ldots,32\}$, c.f., Figure \ref{stable2}, by estimating the minimal injection rate $\rho$ for which the average maximum queue size exceeds 1024.
%
Our results are similar in shape with results observed by Hastad~\etal~in \cite{hastad}, with slight differences in values, due to 
some differences in simulation methodologies. 
Our protocols confirmed their high throughput even on channels without restraints.

In order not to discriminate randomized \algoname{Backoff} protocols,
which may obtain large channel access
peaks from time to time
(unlike our deterministic protocols that ensure bounded channel access at any round), 
in order to evaluate channel restrain we counted an {\em average} number of switched-on stations over rounds.
%
In order to better illustrate bi-criteria comparison of considered protocols,
we compared them with the \algoname{State-aware} protocol, which has full knowledge about all of the queues in the beginning of each round and transmits a packet from a station with the biggest queue. 
This protocol models close-to-optimal queues and channel access for
given injection patterns. 
Figure \ref{EQ2} presents our results in logarithmic scale: 
more efficient protocols in restrain-queue dimensions are closer to the \algoname{State-aware} protocol, what makes \algoname{12 O'clock} adaptive protocol our champion for all injection rates and 8-light \algoname{Interleaved-Selectors} to be the second for injection rates lower than $\rho = 0.3$. 
(Full-sensing version of the \algoname{12 O'clock} protocol has been omitted from the graphs as it behaved similarly
to its adaptive version in our simulations.)

\bibliographystyle{acm}
\bibliography{lib}

\newpage
\begin{appendices}

\section{Conclusions and discussions}\label{Sec:Conc}

We have proposed the $k$-restrained model for multiple access channels and studied throughput and queue stability of deterministic contention resolution protocols. 
%
We have developed protocols with  proven constant upper bound on channel restrain and throughputs $1$ and $1-1/n$, respectively for adaptive and full-sensing protocols. We also show in the Appendix that the presented full-sensing algorithm can be modified to be stable for any throughput smaller than $1$ with the same channel restrain in a cost of larger queue sizes.
\algoname{k-light Selectors} algorithm, though 
achieving smaller throughput, provides a channel restrain solution for a low-level acknowledgment-based-only transmission density.
It was achieved thanks to a newly developed combinatorial tool, $k$-light selector.

Simulations have shown preliminary evidence of performance in the noiseless environment with an adversary imposing asymmetrical queue load to the system (comparing to pure stochastic models).
\algoname{Backoff} protocols were studied as the most commonly used contention-resolution approach. Experiments have repeated \cite{hastad} results in regard of tendencies, with some differences in actual values of measurement, most likely caused by few differences in implementation (as discussed in the Appendix).
\algoname{Backoff} protocols have shown limited throughput, stability, and inability to have small queues and channel restrain at the same time. \algoname{12 O'clock} protocols, on the other hand, have shown stability combined with low channel restrain. 

\section{12 O'clock full-sensing protocol with collision detection}\label{Sec:Full}
12-O'clock full-sensing protocol works similarly to its 
control-bits counterpart, however a decision to change state is based on information who transmitted packet, since adding bits is not allowed for full-sensing protocols. 
To overcome the lack of additional information bits, we implement  precise control mechanism to the ``out-of-order transmissions'' and the collisions enforced by them so that \state{Big} stations could
	be identified. Combined with recognition of ID attached to successful transmissions, any station can learn about \state{Big} station and adjust to it, with some small waste of transmission
	and increased delay.

In our algorithm, typically stations in the \state{Listening} state discover Big stations by reading the ID of station transmitting on the channel and comparing it to the predecessor ID from the list --- if they do not match then the listening station(s) deduct that the transmitter
is in \state{Big} state.
By distinguishing silence from collision, the algorithm is able to manage borderline cases, see the description below. 
However, due to collisions, the protocol is not universally stable, albeit we will prove its stability against injection rates $\rho \le \frac{n-1}{n}$.

\textbf{Technical description.} 
We consider three channel states: $Silence$ when there is no transmission, $Transmission$ when there is single transmission on the channel, $Conflict$ when there is more than one transmission.
Stations can be at one of four states:  \state{Idle},  \state{Listening},  \state{Transmitting} or  \state{Big}. 
The last two states are given the right to transmit; they are distinguished by the order in the list -- only \state{Big} station can transmit out of the order of the list; in the only one possible case when \state{Big} station transmits within the order, collision occurs and later transmissions clarify the system state. The \state{Listening} state is dedicated to listening, while in the Idle state the station neither transmits or listens. We describe these states later in this section.
As previously we assume that transmission happens before the listening phase. Simplified finite state machine for the relationship between those states can be seen on {\it Figure~\ref{12Of:fsm2}}.

\textbf{Initialization.}
In the beginning all but the first two stations are in the \state{Idle} state, while the one with the smallest ID is in \state{Transmitting} state and its successor is in the \state{Listening} state.

\textbf{Idle state.}
In this state the station does not access the channel, it only keeps updating its idling time
until the next wake-up --- it decreases by $1$ each round.
The starting number of idling rounds is either $n$ or $n-1$ or $n-2$, depending on the 
state from which the station switches to \state{Idle} and the message on the channel,
see the description of \state{Listening} and \state{Transmitting} states below.
After awaking, i.e., when the idling time decreases to zero, the state
switches to \state{Listening}.

\textbf{Listening state.} 
A station in the \state{Listening} state considers all three channel state cases,
in the following way. 

Conflict on the channel 
occurs
only when a \state{Big} station $S$ 
interrupted its successor. No information is available on the channel, hence the \state{Listening} station keeps its state unchanged for one more round in order to hear an ID of the \state{Big} station. Note that there will be two stations in the \state{Listening} state and one in the \state{Big} state next round. Both \state{Listening} stations would recognize $S$ as \state{Big} and update their local station lists~accordingly.

Upon hearing a silence,
the \state{Listening} station knows that it will not interrupt a \state{Big} station transmission next round and thus it changes its state to \state{Transmitting}. 

Finally in case of the transmission on the channel, the \state{Listening} station checks transmission ID on the channel and either it takes the token from its successor, or becomes idle and updates the local station list if it was not its predecessor's transmission.
It becomes idle for the next $n-2$, $n-1$ or $n$ rounds until subsequent wake-up; more specifically, the first idling time $n-2$ occurs when station waited additional round after collision on the channel, the second idling time $n-1$ occurs when the station hears a \state{Big} station which is currently located after it on the list of stations, and the last idling time $n$ occurs when the \state{Big} station was located before it on the list.

\textbf{Transmitting state.}  
The \state{Transmitting} state is taken by a station once per cycle in the round corresponding
to its current position on the list of stations, unless there is a \state{Big} station in 
the beginning of that round. 

A station in the \state{Transmitting} state changes its state to \state{Idle} when there is a silence on the channel --- it is possible only when it had no packets and there was no \state{Big} station 
in the beginning of this round. In case of a collision, it updates its local station list by moving its predecessor from the list to the front, as its is the only station which transmission on the channel would allow the \state{Transmitting} station to change its state from \state{Listening} to \state{Transmitting}. 

If a \state{Transmitting} station has successfully transmitted, then there is no \state{Big} station transmission in this round. Additionally, if the \state{Transmitting} station has queue size exceeding $3n$, it changes its state to \state{Big} and keeps transmitting accordingly starting from the next round. 
Otherwise, it changes its state to \state{Idle}, in order to awake in its listening turn during the next cycle, after $n-2$ rounds. 
If the station has not transmitted but a single transmission occurs on the channel, then this is a transmission from predecessor of \state{Big} station (any other \state{Big} station would cause the station not to switch to the \state{Transmitting} state in the first place, as it would switch directly from the \state{Listening} to \state{Idle}) which has not caused a collision only because the \state{Transmitting} station has had no packets to transmit. In this case the station behaves accordingly --- updates the local list of stations and changes its state to \state{Idle}.

\textbf{Big state.} 
At the end of each cycle, a \state{Big} station checks whether its queue size is still
bigger than $2n$; if not, it changes its state to \state{Transmitting}.
In any other round, the \state{Big} station transmits a packet and remains in the same state.
The following property can be easily deducted: once a station changes its state to \state{Big} (which happens when being
in its 
\state{Transmitting} state), it stays there at least till the end of the next cycle;
it may then continue throughout the whole next cycles, until it changes to 
the \state{Transmitting} state at the end of some of them. 

\subsection{Analyses and bounds}

Similarly to the Adaptive protocol analysis, we consider the sum of the queues' sizes in the beginning of a cycle. If it is greater then $\ell = n(3n-1)+1$ we say that it belongs to the $dense$ interval, otherwise it belongs to the $sparse$ interval. This way any execution of the Full-sensing algorithm consists of $dense$ and $sparse$ interleaved intervals.

In relation to a fixed interval we consider the following terminology: station is {\it pre-big} if it had never been in the \state{Big} state and it is {\it post-big} if it was at least once in the Last \state{Big} state, during the considered cycle. Station is {\it potentially-big} if its queue size allows it to become \state{Big} (provided other necessary conditions would hold) or it is in the \state{Big} state.

Each cycle can be only of one of the three types:

\textbf{Type-1.} without any Big station. Token is being passed in the \algoname{Round-Robin} way,
    by adopting \state{Listening} and \state{Transmitting} states. This means that at any single round there is one station in the \state{Transmitting} state and one in the \state{Listening} state. 
    
\textbf{Type-2.} with a Big station $S$ starting to transmit as \state{Big} in some round of the cycle. Here, the token is being passed in the \algoname{Round-Robin} way by applying the sequence of Listening and \state{Transmitting} states to consecutive stations on the list, until $S$ transmits for the second time. The successor of station $S$ cannot recognize $S$ as \state{Big} since $S$ is supposed to transmit by the default \algoname{Round-Robin} way of passing the token within the list order. Conflict occurs if the successor of $S$ has a packet to transmit. Otherwise, in the case of successful transmission, stations in \state{Listening} and \state{Transmitting} states active at this round would read the \state{Big} station ID from the transmission, both changing their states to \state{Idle} afterwards. Otherwise the station in the \state{Transmitting} state learns from the collision about the state of $S$, and then it changes its state to \state{Idle} and updates the local station list. The station which was in the \state{Listening} state at that time learns about the state of $S$ a round after the collision, since 
it
could not be a successor of any \state{Big} station. 
 
\textbf{Type-3.} with a Big station $S$ keeping the ``token to transmit'' for the whole cycle. All but one stations after waking-up in the \state{Listening} state will learn about the state of $S$ and become Idle until the next cycle. One station would not recognize $S$ as \state{Big}, but it will be interrupted by its transmission. Through the collision on the channel it would however learn about the state of $S$, an then it changes its state to \state{Idle} and updates the local station list.

The following two lemmas justify the usage of cycles defined above and provide the limit on the number of collisions.
They will be used implicitly in the analysis. 

\begin{lemma}
Each cycle is of one of the three above types.
\end{lemma}
\begin{proof}
The starting conditions of the algorithm enforce the system to start in the type-1 cycle. 

Type-1 cycle can be followed 
by another type-1 cycle,~if there is no potentially-big station, or by a type-2 cycle~otherwise. 

In a type-2 cycle the Big station is chosen, and therefore it can only be followed by a type-3 cycle --- this is because the \state{Big} station needs to transmit more than $n$ packets in order to start to consider changing its state, which may happen only at the end of some cycle.

A type-3 cycle, with a \state{Big} station keeping the token (to transmit) for the whole cycle, can be followed either by the same type of a cycle if an adversary keeps injecting packets into the \state{Big} station, or by a type-1 cycle if there is no potentially-big station, or by a type-2 cycle otherwise.
\end{proof}

\begin{lemma}
No more than one collision per cycle can occur.
\end{lemma}

\begin{proof}
Note that in a type-1 cycle collision may not occur, as at any single round there is station in the \state{Transmitting} state and another one in the \state{Listening} state.

In a type-2 cycle no collision occurs until the second transmission of a station in the \state{Big} state, by the same reasoning as for type-1 cycles. If the Big station successor has packets in its queues there is a collision on the channel. The station in the \state{Transmitting} state becomes \state{Idle} at the end of this round until the next cycle. Stations further down on the list cannot have the \state{Big} station as predecessor and would wake up in the \state{Listening} state, learn about the \state{Big} station from its transmission and change their state directly to \state{Idle}, hence there can be no more collisions.

A type-3 cycle with a \state{Big} station $S$ keeping the ``token to transmit''. Consider the case, when type-2 cycle precedes. We divide stations of the system into two groups: group-A consists of stations after the \state{Big} station $S$ on the list, which have already learned about the state of $S$ and updated their local lists of stations. Group-B are stations before $S$ on the list, which had no occasion to do so.
If group-A is empty, then there is a single succeeding to $S$ station in the group-B. It causes one collision due to assumption of default \algoname{Round-Robin} predecessor, which is $S$. The rest of the stations in this cycle will switch directly from the \state{Listening} to \state{Idle} state, thus no more collision occur.
If both group-A and group-B are not empty, then no station in the group-B can have $S$ as predecessor, because $S$ is down in the list for any station in group-B by definition, and its not last on their outdated list version since group-A is not empty. Due to group-A stations having their lists updated, $S$ is the first station in their lists, what together with nonempty group-B assumption makes it impossible to any station from the group-A to have $S$ as predecessor. It follows that all of the group-A and group-B stations would change state directly from \state{Listening} to \state{Idle}, thus no collision occur.
If group-B is empty or type-3 cycle precedes the current cycle, than cyclic order of the list does not change (i.e. each station has the same successor and predecessor in the beginning and the end of the cycle), so there is a single succeeding to $S$ station which causes one collision due to assumption of default \algoname{Round-Robin} predecessor, which is $S$. No more stations can have $S$ as predecessor, thus the rest of the stations would change state directly from \state{Listening} to \state{Idle} and no more collision occurs.
\end{proof}

We call a round with collision caused by station in the \state{Big} state an {\em assertion} round. In relation to cycles we assume that there is an assertion round in every cycle, since this is the worst possible case -- no more than one collision in a cycle can occur by Lemma above. By a {\em silent} round we understand any non-assertion round with no successful transmission. 
We say that a station {\em causes} a silent round if during this 
round
it is in state \state{Transmitting}; note that it may occur only if the station has 
empty queue in this round.
Observe also that there cannot be a Big station in a silent round, 
as stations in \state{Big} state have more than $n$ packets in their queues.

\begin{lemma}
In any dense interval, a station can cause a silent round 
at most $n-1$ times while being pre-big.
\label{lemma:pre-big2}
\end{lemma}
\begin{proof}
Silent rounds occur when some station has a ``token to transmit'' but has no packets in its queue. 
%
Assume that a station $S$ has no packets in its queue.
Within dense interval, in each 
round
there is a potentially-big station. For any cycle, if potentially-big station is before $S$ on the list, then $S$ would receive no ``token'' or receive it and decrease its position on the list. The position of $S$ cannot decrease more then $n-1$ times, because there can be no potentially-big station after $S$ if it is the last on the list. 
Since in the dense interval there is always a \state{Big} station, $S$ as the last station in the list either has no possibility to cause silent round (when some other station $S'$ before it in the list changes state to \state{Big}), or becomes \state{Big} itself. Pre-big station life-cycle terminates once station is in the \state{Big} state by our definition, hence through the whole pre-big life-cycle station $S$ may cause no more than $n-1$ silent round.
\end{proof}

\begin{lemma}
In any dense interval, a station causes no silent round while being post-big or in a \state{Big} state.
\label{lemma:post-big2}
\end{lemma}

\begin{proof}
A post-big station $S$ could be in a \state{Big} state, \state{Transmitting} state or in one of the
other two states. In the latter case, it does not attempt to transmit, hence it cannot cause a silent round.
If the station enters \state{Big} state, it switches 
from the \state{Transmitting} state at some 
round
of the cycle, having~more than $3n$ packets in its queue; it'll switch back~to~the~\state{Transmitting} state when having less than $2n$ packets in the end of the cycle.
Thus, in any round of the cycle the number of packets cannot drop below $n$, and
hence no silent round~occurs.

It remains to analyze the case when $S$ is in the \state{Transmitting} state.
Upon leaving the \state{Big} state for the last time, it had at least $n$ packets in its queues and was
placed in the beginning of the list of stations, by the algorithm construction. 
Then, observe that $S$ has had an opportunity to transmit only at some type-2 cycle
when there is no potentially-big station before it on the list or when $S$ is potentially-big at the time it
switches from \state{Listening} to \state{Transmitting} state. In the latter case, 
instead of staying in \state{Transmitting} state it immediately switches to \state{Big} state,
which case we already analyzed in the beginning of the proof.
%
Otherwise (i.e., in the former case), either some potentially-big station after $S$ becomes \state{Big}, 
which implies that in some
of the next cycles, it will switch back to \state{Transmitting} state and the position of $S$ on the list decreases without causing any more
silent rounds, or $S$ receives no ``token to transmit'' and so it cannot cause a silent round by default. 
The position cannot decrease more than $n-1$ times, because there can be no potentially-big station after $S$ if $S$ is the last on the list (the argument is similar to the one from the proof of Lemma~\ref{lemma:pre-big2}). 
Since $S$ had at least $n$ packets when switching from its \state{Big} state,
it can transmit and decrease its position at most $n-1$ times or become \state{Big}, whatever comes first;
in any case, it has at least one packet when being in \state{Transmitting} state.
\end{proof}

\begin{theorem}
\label{theorem:1}
The 12 O'clock full-sensing protocol 
 achieves throughput $1-\frac{1}{n}$ on a channel with restrain of $3$
and the maximum number of packets stored in a round is at most $\ell+(n-1)^2 +n+b = O(n^2+b)$ against leaky-bucket adversary.
\end{theorem}

\begin{proof} 
Injection rate stability limit  of $1-\frac{1}{n}$ follows from inability to identify a Big station $B$ by $B$ station successor in the list. This results in potential collisions every cycle.

The analysis of bounds on the queue size bases upon sparse and dense intervals defined above.
Within a sparse interval, there can be no more then $\ell + n + b$ packets in the stations
at the end of any cycle. Indeed, 
the biggest possible number of packets that the system can start a cycle with is equal to $\ell$, 
and the adversary can inject no more then $n+b$ packets in $n$ consequent rounds of the cycle. 
Once the queue size becomes greater than $\ell$ in the beginning of a cycle, 
the sparse interval terminates and the dense interval begins.

In the remainder, we focus on dense intervals.
Note that in the beginning of a dense interval, the number of packets in the system
is at most $\ell+n$ plus the burstiness above the injection rate (upper bounded by $b$);
indeed, as in the beginning of the preceding cycle the interval was sparse, the number of packets
was not bigger than $\ell$, 
and during that cycle the adversary could inject at most
$n$ packets accounted to the injection rate plus the burstiness.

Within any dense interval, a station in the \state{Big} state is guaranteed to be in each cycle, by the pigeon-hole principle. It makes type-1 cycle impossible to occur. 
Consider type-3 cycles: during those cycles a packet is transmitted in every round, and 
thus a silent round cannot occur; hence the number of packets does not grow (except
of burstiness above the injection rate, but this is upper bounded by $b$ at any round of the interval,
by the definition of the adversary).

In type-2 cycles, post-big stations cannot cause silent rounds, by Lemma \ref{lemma:post-big2},
and stations in \state{Big} state cannot cause silent rounds as they always have more than $2n$ packets
pending.
Hence, type-2 cycles may have silent rounds caused only by pre-big stations. However, there can be no more than $n-1$ pre-big stations in the system in the beginning of the dense interval (because there is at least one potentially-big station). Each pre-big station can cause no more than $n-1$ silent rounds, by Lemma \ref{lemma:pre-big2}. Observe that in each cycle with a silent round some potentially-big station will change its state to Big --- a silent round would not occur if there was a potentially-big station with higher position on the list.
Hence, there can be no more than $n-1$ cycles with silent rounds caused by same (pre-big)
station. 
To summarize, there are at most $n-1$ cycles with silent rounds for each of at most $n-1$ pre-big stations, resulting in the upper bound of $\ell + (n-1)^2 + n+b$ on the sum of the queue sizes in a round.
\end{proof}

\subsubsection{Stability bound improvement}

It was proved in \cite{kow} that it is not possible to construct a full-sensing stable protocol against leaky-bucket adversary $\rho = 1$ for a system with a number of stations bigger than~$3$.
Below we show how the 12 O'clock full-sensing protocol can be modified to withstand higher than $1-\frac{1}{n}$ injection rates:

\begin{lemma}
For any given $\rho < 1$, the 12 O'clock full-sensing protocol can be modified to be stable against the leaky-bucket adversary with injection rate $\rho$ .
\label{theorem:2}
\end{lemma}
\begin{proof}
Algorithm may handle any injection rate  $\rho$ smaller than 1 by following the strategy:

    $\bullet$ \state{Transmitting} station considers itself \state{Big} when it has more then $2n+kn$ packets, where $k\ge \frac{1}{n(1-\rho)}$ is a positive integer;
    $\bullet$ \state{Transmitting} station remembers of being interrupted by its predecessor, and instead of waking up after the subsequent nearly $n$ rounds, as in the original 12 O'clock full-sensing protocol, it  wakes up after $kn$ rounds.
    
This way interruption may happen only once in $kn$ rounds and the adversary with injection rate of $\rho = 1-\frac{1}{kn}$ can be handled. We adjust the sparse/dense border value to $\ell'= n((2+k)n-1)+1$, since the \state{Big} station definition has changed. Following the logic of the Theorem \ref{theorem:1} proof, in the dense interval there are at most $k(n-1)$ cycles for each of at most $(n-1)$ pre-big stations, resulting in the upper bound of total queue sizes of $\ell'+k(n-1)^2 +n+b = O(kn^2+b)$.
\end{proof}

\section{Construction of selector in polynomial time -- \\ proof of Theorem~5.4 }
\label{sec:construction-proof}

First we show that $\cS(n,\omega,k,c)$ is a $k$-light $(n,\omega)$-selector, for sufficiently large constant $c>0$.
First we consider more complex case of $n>a|W|\alpha$.
Let set $A\subseteq V$ be of size~between $\omega/2$ and $\omega$.
Suppose, to the contrary, that there are less than $\omega/4$ elements in $A$ hit by some sets in $\cS(n,\omega,k,c)$.
It implies that there is a subset $B\subseteq A$ of 
size $\omega/4+1$ such that none of the elements in~$B$
is hit by sets from $\cS(n,\omega,k,c)$.

\noindent
\textsf{Claim:}
Every $w\in N_G(B)$ has more than 
$c\delta$ neighbors in~$A$, 
where $c\delta$ is a disjunctness parameter of $\cM$.

The proof is by contradiction.
Assume, for simplicity of notation, that $w\in W$ 
is the $w$-th element of set~$W$.
Suppose, to the contrary, that
there is $w\in N_G(B)$ which has
at most $c\delta$ neighbors in~$A$,
that is, $|N_G(w)\cap A|\le c\delta$.
By the fact that $\cM$ is a $c\delta$-disjunct  
superimposed code, for $a=O((c\delta)^2 \log^2 n)$,
we have that, for every $v\in N_G(w)\cap A$, the equalities
\[
F_{w\cdot a+y}\cap A = 
(M_y\cap N_G(w)) \cap A =
M_y\cap(N_G(w)\cap A) =
\{v\} 
\]
hold, for some $1\le y\le a$.
This holds in particular for every $v\in B\cap N_G(w)\cap A$.
There is at least one such~$v\in B\cap N_G(w)\cap A$, 
because set $B\cap N_G(w)\cap A$ is nonempty 
due to $w\in N_G(B)$ and $B\subseteq A$.
The existence of such~$v$ is in contradiction with the choice of set~$B$.
More precisely, $B$ contains only elements which are not hit
by sets from $\cS(n,\omega,k,c)$, but $v\in B\cap N_G(w)\cap A$ 
is hit by some set $F_{w\cdot a+y}$, thus is also hit by some children $S_j\in \cS(n,\omega,k,c)$ of $F_{w\cdot a+y}$.
This makes the proof of the Claim complete.

Recall that $|B|=\omega/4+1$.
By dispersion, the set $N_G(B)$ is of size larger than $(1-\ep)|W|$, 
hence, by the Claim above, the total number of edges 
between the nodes in $A$ and 
$N_G(B)$ in graph $G$ is larger than 
\[
(1-\ep)|W|\cdot c\delta = 
(1-\ep)\Theta((\omega/4+1)d/\delta) \cdot c\delta >
\omega d \ ,
\]
for a sufficiently large constant~$c$. 
This is a contradiction, since the total number of edges
incident to nodes in~$A$ is at most $|A|\cdot d=\omega d$. 
It follows that $\cS(n,\omega,k,c)$ is a $k$-light $(n,\omega,k)$-selector, for a sufficiently large constant $c$.

Before estimate the length $m$ of this selector, note that the union of all sets $F_j$ in the case $n> a|W|\alpha$
is at most $a\cdot (|V|\cdot d)$, because an element in some $F_j$ corresponds to some edge in the disperser
and repeats at most as many times as the number of rows $a$ in the superimposed code $\cM$.
Hence, the amortized number of children $S\in \cS(n,\omega,k,c)$ of a set $F_j$, parameter $\alpha$, is 
at most 
\[
\frac{a\cdot (|V|\cdot d)}{k} \cdot \frac{1}{a|W|} + 1 
\ .
\]

The length $m$ of this selector is thus at most
\begin{eqnarray*}
\min\{n,a|W|\alpha\} 
&=&
O\left(\min\left\{n,\delta^2\log^2 n\cdot
\omega d/\delta  + \frac{nd \cdot (c\delta)^2 \log^2 n}{k} \right\}\right)\\
&=&
O\left(\min\left\{n,(\omega + n/k)\polylog n\right\}\right)
\ ,
\end{eqnarray*}
since $d=O(\polylog n)$ and $\delta=O(\log^3 n)$.

The case $n\le a|W|\alpha$ is clear, since each element $i$ in a set~$A$ of
size~between $\omega/2$ and $\omega$ occurs as a singleton in some selector's set, mainly in~$S_i$.

\section{Algorithms simulations}\label{Sec:Sim}
In order to evaluate efficiency of developed protocols, we performed simulations for both new and existing algorithms and compared the results. We analyzed the impact of the execution length, system size and injection rates on the queue sizes and channel restrain.
Experiments were limited to a setting strictly based upon the model.

We collate Adaptive and Full-sensing versions of the \algoname{12-O'clock} algorithm with Backoff exponential and polynomial algorithms. Additionally, we take into account acknowledgment-based algorithms: \algoname{Round-Robin} and $8$-light \algoname{Interleaved-Selectors}.
Our main simulation goals are to analyze and compare across the
considered protocols: 

\textbf{(i) General workflow} for stable injection rates;

\textbf{(ii) Critical injection rates}, that is - the lowest injection rates where queue size or latency show dependency on the number of rounds passed (because practically time-dependent behavior indicates instability);

\textbf{(iii) Channel restrain} below critical injection rates, so that channel restrain in stable executions could be evaluated.

A digest of the obtained results is presented in Figures~\ref{QbyRa}-\ref{EQ}.
Experiment results are presented without error bars to improve clarity, as several graphs are present on each figure. 
Each recorded result is an average of 120 experiments of one million rounds each.

\subsection{Implementation}

We have implemented \algoname{12 O'clock} adaptive and full-sensing versions, 8-light \algoname{Interleaved-Selectors}, \algoname{Round-Robin}, \algoname{Backoff} exponential, linear and square polynomial versions in Java and Julia languages. Their relationship to colors used on plots in Figures~\ref{QNbyRo}-\ref{EQ} is described in 
Figure~\ref{LP}.

\paragraph*{\bf Backoff protocols}
In general we follow the model from \cite{hastad}. That is, 
we use synchronous model
with a message length limited to the length of the transmission phase in a single round, and 
we do not terminate undelivered messages. However, we 
employ 
an upper bound on Backoff counter, as in real-world applications.
A station attempts to transmit and learns about transmission success or failure within a single round. 
Analogously to the description in \cite{hastad}, 
the limits on the Backoff contention window size function (without window reset) are
$2^i$, $2i$, $2i^2$ accordingly for exponential, linear and square polynomial
versions of Backoff protocol, where $i$ denotes the unsuccessful transmission counter. The size of the window is limited by $2048$, as the biggest system size studied in this work is equal to $32$. That allows to protect protocols from
unnecessary increase of the window size and thus 
improves their worst-case stability.
\paragraph*{\bf Acknowledgment-based protocols}
\algoname{Round-Robin} protocol allows any station $i$ to transmit alone in rounds
$i$ modulo~$n$.
$8$-light \algoname{Interleaved-Selectors} are based on randomly generated binary matrices, tested to satisfy the definition
of $k$-light $(n,\omega)$-selector.
%

\paragraph*{\bf Adversary}
An adversary is defined by three parameters used at each round $r$: injection rate $\rho$ -- the probability that an adversary will have one more packet in its stock, burst-probability $p$ -- the probability of adversary making a decision to inject all of the stock packets
at once, and finally the stock size limit $b$ -- a constant forcing the adversary to inject all of its stock packets once the stock size is equal to $b$. For each packet decided to be injected to the system,
the adversary selects a station $S_i$ for injection with probability $P_i$, where $i \in \{1, 2, 3, \ldots, n\}$: $P_1 = P_2 = \frac{1}{3} + \frac{1}{3n}$; $P_{i > 2} = \frac{1}{3n}$.

Injection rate $\rho$ and burst-parameter $p$ have values in $(0,1)$.
The burst-probability models the adversary injection behavior: between rare bursts of large numbers of packets (close to 0) and steady flow (close to 1). The stock-size $b$ is a constant equal to $256$, basing on operational buffer size limits. 
After performing some preliminary experiments for different values of $p$,
we have chosen $p=0.5$ for this presentation -- it occurred not to influence the system performance as much as we had expected.

\subsection{Boundaries on stable injection rates}

\subsubsection{Boundaries measurement}

\begin{figure*}[t!]
    \begin{subfigure}[t]{0.5\textwidth}
	\includegraphics[height=1.16in]{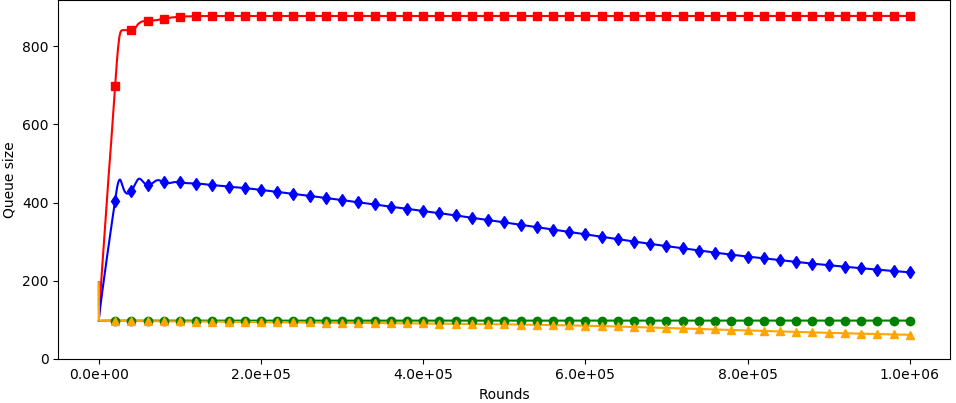}
	\caption{\algoname{12 O'clock} full-sensing protocol queues against injection rates $\rho = 0.968$, started 
	with queues in each 
	station equal to $q=96$.}
	\label{QbyRa}
    \end{subfigure}%
    ~ 
    \begin{subfigure}[t]{0.5\textwidth}
    \includegraphics[height=1.16in]{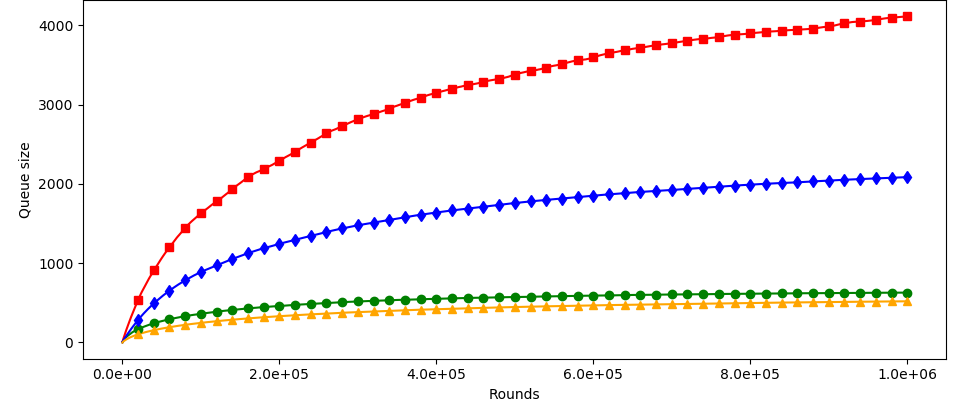}
	\caption{Backoff exponential protocol queues when run against injection rates $\rho = 0.9$, starting with empty starting queues.}
	\label{QbyRb}
    \end{subfigure}
    ~
    \begin{subfigure}[t]{1 \textwidth}
	\includegraphics[height=0.5in]{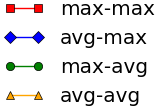}
	\caption{Relation between markers, colors and measurements of queues}
	\label{LM}
    \end{subfigure}
    \caption{Protocols
	 during 1 mln rounds for a system size $32$ against flat packet distribution.}
\end{figure*}
We took into consideration several measurements of queues of a protocol (at round $r$): 
a maximal queue size of a single station occurring up to round $r$ (max-max); an average, taken over $r$ rounds, of a maximal queue size of stations at a round 
(avg-max); a maximal over $r$ rounds of an average queue size of all  stations at a round 
(max-avg); and finally, an average over $r$ rounds of an average queue size of all stations in a round (avg-avg). 
Note that the max-max and max-avg measurements can not decrease and are always divergent against an adversary without burstiness limit (with probability $1$).

Comparison of those measurements for \algoname{12 O'clock} full-sensing and \algoname{Backoff} exponential protocols is shown in Figures~\ref{QbyRa} and \ref{QbyRb} for system size $n=32$ against flat packet distribution. 
The \algoname{12 O'clock} full-sensing protocol started with $3n=96$ packets per station (i.e., the total system queue size equal to $3n^2$) stabilizes against injection rate $\rho = 0.968$, which is slightly smaller than the theoretical stability boundary $\rho = \frac{31}{32} = 0.96875$, for all four measurements and its both avg-max and avg-avg measurements decrease after handling the starting queues burst
(Figure \ref{QbyRa}). 
The adaptive version of the protocol behaves in a similar way when started with full queues against any injection rate $\rho < 1$ and is 
also bounded against injection rate $\rho = 1$. 

For exponential \algoname{Backoff} protocol,
c.f., Figure \ref{QbyRb}, 
the max-max and avg-max measurements seem to diverge, while
the other two measurements converge but to much higher values
than in case of \algoname{12 O'clock} full-sensing protocol.

Based on the above results, we have chosen the avg-max measurement for further comparison of protocols. This is bacause when considering
other three ways of measuring: max-max is highly volatile for the randomized protocols (and thus it would not be fair for comparison
randomized and deterministic protocols) while avg-avg and max-avg 
do not envision 
the worst case scenario we are focused on in this work.
Note that the avg-avg measurement, studied in \cite{hastad} 
and in many other previous papers considering stochastic injections, may yield stability while having single queues
many times above the studied average.

\subsubsection{Boundaries for system sizes $n \in \{4,5,\ldots,32\}$}

In order to see how system queues behave for different system sizes, we have combined simulation results for system sizes $n \in \{4,5,\ldots,32\}$ on a single plot (Figure \ref{QNbyRo}). 

We have excluded the full-sensing version of \algoname{12 O'clock} since its results are similar to the adaptive version in most of the considered scenarios. In this section we discuss the combined boundaries in Figure \ref{QNbyRo}
and the stable injection rates depicted in Figure \ref{stable} defined as minimal injection rates $\rho$ for system size $n$ required to make the value of avg-max measurement to exceed the constant value $\delta=1024$.

Our results are similar in shape with results from \cite{hastad}, with difference in values. It could be explained by the following:
we implemented more adversarial behavior instead of Poisson distribution, used 1 million instead of 10 millions iterations for experiment length, avg-max measurement instead of avg-avg (to 
better capture worst-case behavior), and finally we set-up a maximal window size limit to comply with real applications of Backoff. Specifically, the maximal window size limit improves the efficiency of exponential \algoname{Backoff} protocol in comparison to other versions of \algoname{Backoff} protocols in our context.

\begin{figure*}[t!]
    \begin{subfigure}[t]{0.5\textwidth}
\includegraphics[height=1.63in]{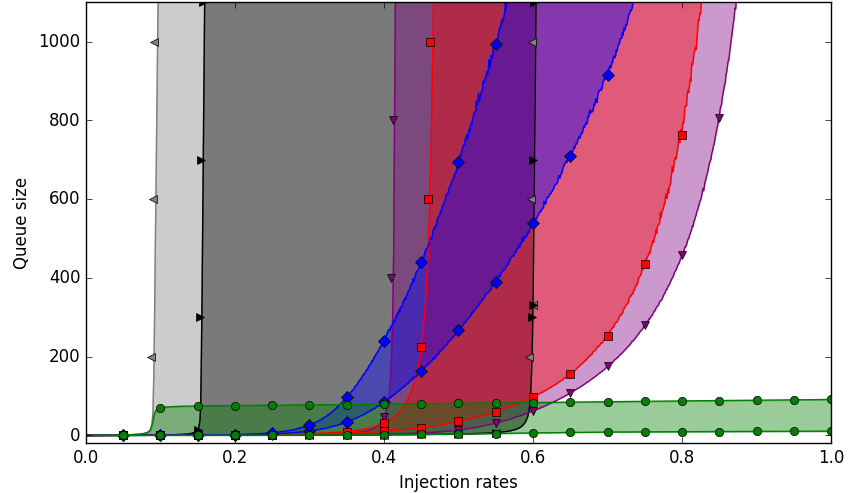}
  \caption{%
  	Queue size by injection rates 
  	$\rho \in \left[0,1\right]$.}
  \label{QNbyRo}
    \end{subfigure}%
    ~ 
    \begin{subfigure}[t]{0.5\textwidth}
	\includegraphics[height=1.63in]{stability_bound.png}
	\caption{Stable injection rates by system size.
		} 
	\label{stable}
    \end{subfigure}
    ~
    \begin{subfigure}[t]{1 \textwidth}

	\includegraphics[height=0.5in]{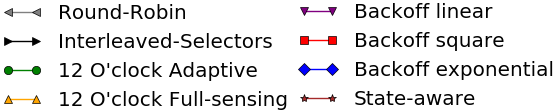}
	\caption{Relation between markers, colors and protocols.}
	\label{LP}
    \end{subfigure}
    \caption{System size influence over queue size for system size $n \in \{4,5,\ldots,32\}$.}
\end{figure*}

\algonamebf{Acknowledgment-based} protocols have the same \algoname{Round-Robin} implementation of selectors for system sizes $n \in \{4,5,\ldots,15,17,18\}$, because we were unable to generate better $(n,\omega)$-selectors for $\omega \leq n/2$ required for \algoname{Interleaved-Selectors} 
in
those cases. It follows that their plots overlap. The best achieved stability bound 
is around $\rho = 0.6$ for system size $n=4$,
and it gradually decreases with the increasing system size
	(in a pace resembling hyperbola). 
On the other hand, we can observe an improvement of \algoname{Interleaved-Selectors} over \algoname{Round-Robin} protocol for bigger systems: 
for some system sizes its stability range is even a few times bigger
than the stability range of \algoname{Round-Robin}.
The irregular shape of \algoname{Interleaved-Selectors} 
stable injection rates in Figure~\ref{stable} is 
caused by selectors 
being generated independently for each (larger) system size,
which leaves a clear scope for further optimization 
	of the quality of selectors.


%
\textbf{\algonamebf{Backoff}} protocols are system size dependent, and the following two interesting phenomena can be observed. First, the lower rank polynomial/function of  \algonamebf{Backoff} protocol the wider extremes in stable injection rates it achieves for different system sizes, e.g., 
$\rho \in [0.55, 0.7]$ for exponential version versus $\rho \in [0.45, 0.8]$ for square and $\rho \in [0.4, 0.85]$ for linear version, c.f., the values of $\rho$ at the top boundaries of corresponding regions in Figure~\ref{QNbyRo}. The second observation is that for smaller system sizes the protocols with lower rank function achieve higher stable
injection rates while for larger systems (starting from some size specific for the considered functions) the tendency is opposite
c.f., Figure~\ref{stable}.%

\remove{%
\textbf{\algonamebf{Backoff} polynomial} protocols show better results for smaller systems. Lower rank polynomial protocol has greater extremes in comparison to its square counterpart: $\rho \in [0.4, 0.85]$ against $\rho \in [0.45, 0.8]$ for square version. We 
hypothesize that this tendency would continue for higher ranks of polynomial protocols.
}

\textbf{\algonamebf{12 O'clock} adaptive} protocols show the most system-size independent behavior, with \algoname{12 O'clock} adaptive protocol being a champion in terms of queue size stability, c.f., Figure~\ref{stable}. Note also that the stable injection rates of \algoname{12 O'clock} full-sensing protocol improve with increasing system size.

\remove{

\subsubsection{Boundaries comparison -- conclusions}

In general reviewed protocols has better results for smaller system sizes.

8-light \algoname{Interleaved-Selectors} stable injection rates vary, being dependent on quality of generated selectors and their adaptation to the current adversary strategy. Which leaves a possibility for improvement by providing better selectors.

\algoname{Backoff} protocols  of higher polynomial rank improve stability boundary for big systems in cost of lowering it for small system size, hence they are system size dependent.

\algoname{12 O'clock} adaptive algorithm is the only stable protocol of all with queue size limited all the time for all system sizes with full-sensing version following it.

}

\subsection{Channel restrain and stability}

\begin{figure*}[t!]
    \begin{subfigure}[t]{0.5\textwidth}
\includegraphics[height=1.47in]{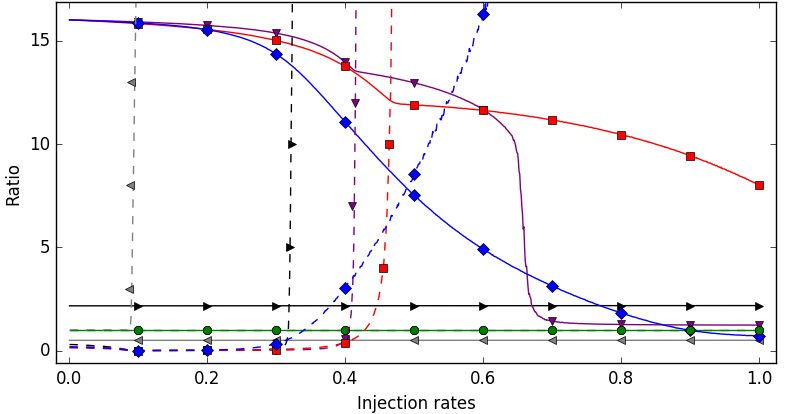}
\caption{Ratios of the average round channel accesses 
		(solid lines) and the 
		queue size (dotted lines) of a protocol to 
	the corresponding performance of \algoname{12 O'clock} adaptive protocol.} 
	
  \label{R}
    \end{subfigure}%
    ~ 
    \begin{subfigure}[t]{0.5\textwidth}
\includegraphics[height=1.5in]{e_q_a2.png}
\caption{Average round channel accesses against  
	queue sizes in logarithmic scale,  with markers set every $0.1$.}
  \label{EQ}
    \end{subfigure}
    \caption{Average round channel access against queue sizes for injection rates $\rho \in \left[0,1 \right]$ and system size $n=32$.}
\end{figure*}

In order not to discriminate randomized \algoname{Backoff} protocols,
which may obtain large channel access
peaks from time to time
(unlike our deterministic protocols that ensure bounded channel access at any round), we 
count how many stations were switched-on on average (over rounds) to evaluate channel restrain. 
In Figure~\ref{R} we show the ratios of channel accesses and queue size of the considered protocols to the corresponding performances of  
	\algoname{12 O'clock} adaptive protocol.

Observe that for \algoname{Backoff} protocols the channel access levels are relatively high and close to the system size when these protocols work within their stable boundaries. 
In contrary, the \algoname{12 O'clock} and \algoname{Acknowledgment-based}  protocols channel access is low and upper bounded by a constant (i.e., independently on all system parameters).

In order to better illustrate bi-criteria comparison of protocols,
we compare them with the \algoname{State-aware} protocol, which has full knowledge about all of the queues in the beginning of each round and transmits a packet from a station with the biggest queue. 
This protocol models close-to-optimal queues and channel access for
given injection patterns. 
Figure \ref{EQ} presents our results in logarithmic scale: 
more efficient protocols in restrain-queue dimensions are closer to the \algoname{State-aware} protocol, what makes \algoname{12 O'clock} adaptive protocol our champion for all injection rates and 8-light \algoname{Interleaved-Selectors} to be the second for injection rates lower than $\rho = 0.3$. 
(Full-sensing version of the \algoname{12 O'clock} protocol has been omitted from the graphs as it behaved similarly
to its adaptive version in our experiment.) 

\end{appendices}

\end{document}